\documentclass[a4paper,UKenglish,cleveref,autoref,thm-restate]{lipics-v2021}
\hideLIPIcs

\title{Censorship-Resistant Sealed-Bid Auctions on Blockchains}

\author{Orestis Alpos}{Common Prefix}{}{}{}
\author{Lioba Heimbach}{Category Labs}{}{}{}
\author{Kartik Nayak}{Duke University}{}{}{}
\author{Sarisht Wadhwa}{Duke University}{}{}{}

\authorrunning{O. Alpos, L. Heimbach, K. Nayak, and S. Wadhwa}

\Copyright{Orestis Alpos, Lioba Heimbach, Kartik Nayak, and Sarisht Wadhwa}



\relatedversion{}

\graphicspath{{./figures/}}

\usepackage{xurl}
\usepackage{amsfonts,bm,bbm,nicefrac,comment,float,thmtools,thm-restate,multirow}
\usepackage[export]{adjustbox}
\usepackage{xspace}
\usepackage[inline]{enumitem}
\usepackage{mathtools}
\usepackage{algorithm}
\usepackage[noend]{algpseudocode}
\usepackage{eqparbox}
\usepackage{xparse}
\usepackage{pifont}
\usepackage{booktabs}
\usepackage{wrapfig}

\usepackage{tikz}
\usetikzlibrary{arrows.meta,positioning,shapes,fit,calc}
\usetikzlibrary{shapes.geometric}

\usepackage{colortbl}
\usepackage[table]{xcolor}

\renewcommand{\paragraph}[1]{\medskip\noindent\textbf{#1}\enspace}

\newlength{\AlgCommentStart}
\algnewcommand{\LComment}[1]{%
  \Statex
  \hspace*{\AlgCommentStart}%
  $\triangleright$~%
  \parbox[t]{\dimexpr\linewidth-\AlgCommentStart\relax}{\footnotesize #1\strut}%
}

\algrenewcommand{\algorithmiccomment}[1]{%
  \hfill $\triangleright$ \eqparbox{COMMENT\thealgorithm}{\footnotesize #1}%
}

\algnewcommand{\algrule}[1][.2pt]{\par\vskip.5\baselineskip\hrule height #1\par\vskip.5\baselineskip}
\algnewcommand{\algdescription}[1]{\Statex \footnotesize{\textbf{Description: }#1}}
\algnewcommand{\algincentives}[1]{\Statex \footnotesize{\textbf{Incentive Structure: }#1}}

\algnewcommand\algorithmicoffline{\textbf{offline}}
\algdef{SE}[OFFLINE]{Offline}{EndOffline}[2]{%
  \algorithmicoffline\ \textsc{#1}(#2)}{}
\algtext*{EndOffline}

\newcommand{\CN}{\ensuremath{\mathbb{N}}\xspace}

\newcommand{\tx}{\ensuremath{{tx}}\xspace}

\newcommand{\slot}{\ensuremath{s}\xspace}

\newcommand{\txfee}{\ensuremath{{\varphi}}\xspace}
\newcommand{\userset}{\ensuremath{\mathcal{U}}\xspace}
\newcommand{\user}{\ensuremath{{u}}\xspace}
\newcommand{\depositamount}{\ensuremath{{D}}\xspace}

\newcommand{\Hash}{\ensuremath{\textsc{Hash}}\xspace}
\newcommand{\gossipdelay}{\ensuremath{\Delta_g}\xspace}
\newcommand{\syncdelay}{\ensuremath{\Delta}\xspace}

\newcommand{\timestampers}{\ensuremath{\mathcal{P}_{ts}}\xspace}
\newcommand{\numtimestampers}{\ensuremath{n_{ts}}\xspace}
\newcommand{\ftimestampers}{\ensuremath{f_{ts}}\xspace}
\newcommand{\timestamper}[1]{\ensuremath{P_{\textit{ts,#1}}}\xspace}
\newcommand{\timestamperc}{\ensuremath{P_{\textit{ts}}}\xspace}
\newcommand{\ILproposer}{\ensuremath{P_{il}}\xspace}
\newcommand{\ILproposers}{\ensuremath{\mathcal{P}_{il}}\xspace}
\newcommand{\numILproposers}{\ensuremath{n_{il}}\xspace}
\newcommand{\ILthreshold}{\ensuremath{f_{\textit{il}}}\xspace}

\newcommand{\thiding}{\ensuremath{t_h}\xspace}

\newcommand{\adversary}{\ensuremath{\mathcal{A}}\xspace}

\newcommand{\commitment}{\ensuremath{\mathsf{cm}}\xspace}
\newcommand{\clocktime}{\ensuremath{t}\xspace}

\newcommand{\certificate}{\ensuremath{\sigma}\xspace}
\newcommand{\timestamp}{\ensuremath{\tau}\xspace}
\newcommand{\medianTimestamp}{\ensuremath{\tau}\xspace}
\newcommand{\auctiontime}{\ensuremath{t}\xspace}
\newcommand{\numauctions}{\ensuremath{n_{A}}\xspace}
\newcommand{\secret}{\ensuremath{\rho}\xspace}
\newcommand{\depositcontract}{\ensuremath{\mathcal{DC}}\xspace}
\newcommand{\tstart}{\ensuremath{\auctiontime_{\mathrm{start}}}\xspace}
\newcommand{\tend}{\ensuremath{\auctiontime_{\mathrm{end}}}\xspace}
\newcommand{\rootslot}{\ensuremath{\mathsf{root}_{\slot}}\xspace}
\newcommand{\handle}{\ensuremath{h_\user}\xspace}
\newcommand{\memberproof}{\ensuremath{\pi^{\textsf{m}}_\user}\xspace}
\newcommand{\idproof}{\ensuremath{\pi^{\textsf{a}}_\user}\xspace}

\newcommand{\setbids}{\ensuremath{\mathcal{B}}\xspace}
\newcommand{\bid}{\ensuremath{b}\xspace}
\newcommand{\val}[1]{\ensuremath{\mathrm{val}(#1)}\xspace}
\newcommand{\useractivity}{\ensuremath{\mathbf{U}}\xspace}

\newcommand{\auctioneer}{\ensuremath{a}\xspace}

\newcommand{\auctionid}{\ensuremath{\mathsf{AucID}}\xspace}
\newcommand{\itemz}{\ensuremath{\mathsf{item}}\xspace}
\newcommand{\Addr}{\ensuremath{\mathsf{A}}\xspace}

\newcommand{\winningbid}[1]{\ensuremath{b^{#1}_{w}}\xspace}

\newcommand{\bindingProperty}{No Free Bid Withdrawal\xspace}

\newcommand{\variable}[1]{\ensuremath{\mathsf{#1}}\xspace}
\newcommand{\varDeposits}{\variable{deposits}}
\newcommand{\varInactives}{\variable{inactives}}
\newcommand{\varSlashed}{\variable{slashed}}
\newcommand{\varMinDeposit}{\variable{minDeposit}}
\newcommand{\varPath}{\variable{path}}
\newcommand{\varStatement}{\variable{statement}}
\newcommand{\varWitness}{\variable{witness}}
\newcommand{\varCounter}{\variable{ctr}}
\newcommand{\varPerSlot}{\variable{aucsPerSlot}}
\newcommand{\varItem}{\variable{item}}
\newcommand{\varAllAuctions}{\variable{auctions}}
\newcommand{\varNOP}{\variable{NOP}}

\newcommand{\function}[1]{\ensuremath{\textsc{#1}}}
\newcommand{\proveEligibility}{\function{elgPrv}}
\newcommand{\verifyEligibility}{\function{elgVrf}}
\newcommand{\zkprove}{\function{ZkPrv}}
\newcommand{\zkverify}{\function{ZkVrf}}

\newcommand{\submitCom}{\function{submitCom}}
\newcommand{\signCom}{\function{signCom}}
\newcommand{\proveInactivity}{\function{proveNOP}}
\newcommand{\revealBid}{\function{revealBid}}
\newcommand{\updateBid}{\function{processBid}}
\newcommand{\deposit}{\function{deposit}}
\newcommand{\registerDeposit}{\function{registerDeposit}}
\newcommand{\receiveTimestamp}{\function{receivedTimestamps}}
\newcommand{\updateActivity}{\function{updateActivity}}
\newcommand{\updateHistory}{\function{updateHistory}}
\newcommand{\makeOffline}{\function{makeInactive}}
\newcommand{\slashUsr}{\function{slashUsr}}
\newcommand{\withdraw}{\function{withdraw}}

\newcommand{\negl}{\mathsf{negl}}
\newcommand{\ape}{\ensuremath{\epsilon}\xspace}

\begin{document}

\nolinenumbers

\maketitle

\begin{abstract}
Auctions are now central to blockchain markets, settling NFT sales, token launches, DeFi liquidations, and arbitrage opportunities. Each on-chain bid is a public transaction whose inclusion is decided by a single consensus proposer per block. The proposer can observe pending bids, exclude competitors, and submit bids of their own, breaking the fairness guarantees of classical sealed-bid auctions.

To enable latency-sensitive sealed-bid auctions in blockchain settings, we formalize four properties -- each necessary to prevent a concrete attack -- and design a protocol achieving all four: hiding bid contents, existence, and bidder identity until reveal (\emph{Hiding}); counting all timely honest bids and rejecting late adversarial bids (\emph{Simultaneous Release}); preventing silent withdrawal of committed bids (\emph{No Free Bid Withdrawal}); and charging on-chain fees only to winners (\emph{Auction Participation Efficiency}).

Our protocol uses a timestamping oracle (instantiated with a committee of  $2\ftimestampers+1$ timestampers) and a censorship-resistant inclusion predicate (instantiated using a FOCIL-based inclusion list), with only the winning bid settled on-chain. Our construction relies on two zero-knowledge proofs: an \emph{eligibility proof} that anonymously proves deposit membership to the timestamping committee, and an \emph{auction proof} that binds a bid to a specific auction for the inclusion list committee. We implement both using Groth16 over BN254 with Poseidon hashing in arkworks/Rust: the auction proof generates in 13\,ms and verifies in under 1\,ms; eligibility proofs for Merkle trees up to $2^{32}$ bidders generate in 47--159\,ms and verify in ${\sim}1$\,ms. 

Together, this yields a sealed-bid auction primitive practical for high-value, time-sensitive blockchain settings.
\end{abstract}

\section{Introduction}
Auctions decide a growing share of outcomes in blockchain markets, from NFT and token sales to blockspace allocation, DeFi liquidations, and arbitrage capture. Yet implementing auctions on a blockchain introduces a unique structural challenge: every bid passes through a single consensus proposer per slot, who can observe pending bids, decide which competitors to include, and add its own bid to the auction. This proposer monopoly has produced a multibillion-dollar MEV economy~\cite{daian2020flashboys} and is most damaging in latency-sensitive, high-value sealed-bid auctions, where a single proposer's actions in one slot determine the outcome.

The proposer's monopoly creates several distinct failure modes for on-chain sealed-bid auctions. Pending bids sit in the (public) mempool, where the proposer and other observers can read their values, identities, and submission times before inclusion. The proposer can then censor honest bids, condition its own bid on those it sees, and bid with a latency advantage over other bidders. An on-chain auction implemented as a smart contract does not address these failures: the contract sees only the bids the proposer has included, leaving consensus-layer censorship and the proposer's ability to insert its own bid after observing others (i.e., a last-look advantage) untouched. Commit-and-reveal mitigates visibility by hiding bids until a reveal phase, but it introduces a different problem: a bidder can withhold their reveal once others' bids leak~\cite{ens1}. Even absent any manipulation, a fully on-chain auction requires every bid to be submitted as a transaction. These transactions pay inclusion fees regardless of the auction outcome, so losing bidders incur a cost not present in traditional sealed-bid settings. The damage is sharpest in latency-sensitive, high-value auctions such as DeFi liquidations, batch auctions, and time-bounded arbitrage windows. A single slot of proposer-induced delay (around 12 seconds on Ethereum) is enough for asset prices to move materially: a delayed liquidation may settle against a different reference price, a batched auction may exclude a bid that would have moved the uniform clearing price, and an arbitrage opportunity may evaporate before the next slot.

An effective blockchain sealed-bid auction protocol must therefore satisfy four properties: three that defend the auction against adversarial proposers and bidders, and one that aligns its payoff structure with classical auctions.
\begin{enumerate}[leftmargin=*,nosep]
\item \textbf{Hiding.} Until the auction deadline, in an ideal world, no one learns \emph{what} a given user bid, \emph{when} a bid was submitted, or even \emph{whether} a user bid. Without hiding, an attacker (such as a late-arriving bidder, or the next block proposer) can condition their own bid on observed bids: they outbid only when the asset's value exceeds the observed bid and walk away otherwise, effectively diminishing the auction~\cite{sealedbid25}.
\item \textbf{Simultaneous Release.} In an ideal world, the auction counts every honest bid submitted before the deadline and rejects every bid after it. Without this, a proposer can censor all rival bids and submit a near-zero bid of their own, winning the item essentially for free~\cite{sealedbid25}.
\item \textbf{No Free Bid Withdrawal.} Once a user submits a bid, they cannot retract it without cost. Without this, in a commit-and-reveal scheme, a bidder can wait for other bids to leak during reveal and refuse to open their own commitment whenever doing so is unfavorable. Adverse selection then re-enters through the back door.
\item \textbf{Auction Participation Efficiency (APE).} On-chain auctions should match the payoff dynamics of their classical counterparts, i.e., only the winning bidder pays. Without APE, every bidder pays an inclusion fee regardless of outcome, so losing bidders earn negative expected utility from participating.
\end{enumerate}

The first three properties form the security backbone of the platform, while the fourth property is an efficiency property. Together, they recover the structural assumptions under which sealed-bid auctions are analyzed in the classical auction-theory literature. An ideal platform achieving all four properties thus places the auction in the same strategic setting assumed by Vickrey~\cite{vickrey1961counterspeculation}, Myerson~\cite{myerson1981optimal}, Riley and Samuelson~\cite{riley1981optimal}, and Krishna~\cite{krishna2009auction}, so that, conditional on the standard economic assumptions of the value environment also holding (e.g., independent private values, risk-neutral bidders, no collusion), the classical equilibrium and revenue results apply.

Achieving all four properties simultaneously is hard, and prior work has addressed only subsets. Some proposals remove the single-proposer monopoly by restructuring consensus~\cite{sealedbid25} to achieve certain properties; however, they still fail to achieve all the above, notably lacking APE since every bidder must still pay on-chain inclusion fees. Lighter approaches add inclusion~\cite{soispoke2024focil,wadhwa2025aucil} or commit-and-reveal~\cite{Riggs,zeroauction,novakovic2024cryptobazaar} mechanisms on top of existing chains, but lack hiding and cannot guarantee inclusion in narrow, time-bounded windows. None of the previous works are able to reduce the cost to losing bidders and achieve APE. 

\paragraph{Protocol at a glance.}
We close this gap with a lightweight protocol built on two modular primitives: a \emph{timestamping oracle} that attests bid creation times, and a \emph{censorship-resistant inclusion scheme} that admits valid bids before the deadline and rejects late ones. We instantiate the oracle with a $2\ftimestampers+1$ timestamping committee and the inclusion scheme with a FOCIL~\cite{soispoke2024focil}-based mechanism; any other constructions providing the same guarantees would suffice. The protocol runs two preparatory phases ahead of any auction and two per-auction phases; \cref{fig:overviewt} gives an overview. In \emph{Pre-Phase~A}, an auctioneer registers an auction with a deposit contract. In \emph{Pre-Phase~B}, each user makes a one-time deposit and posts a commitment $C_u$ to a private secret $\secret$, enabling later anonymous proofs of registration. The same deposit and commitment cover the user's participation in any number of subsequent auctions; the locked deposit underpins \emph{No Free Bid Withdrawal}, since a user who fails to reveal a committed bid forfeits it.

For each auction, the protocol then runs two phases. In \emph{Phase~1 (timestamping)}, a user gossips a hiding commitment to its bid, paired with an anonymous proof of registration, to the \emph{timestampers}. The user aggregates the timestampers' signed responses into a single timestamp certificate that attests the bid was generated before~$\tend$. The commitment delivers \emph{Hiding}, assuming that the network preserves sender anonymity~\cite{fanti2018dandelion++}.  In \emph{Phase~2 (bid submission)},  the user reveals the bid to an inclusion scheme which ensures censorship-resistance of the reveal. We achieve this by using a modified FOCIL~\cite{soispoke2024focil} based scheme, paired with the verification of timestamp certificate received from the Phase~1. The inclusion scheme ensures that every bid backed by a valid pre-$\tend$ certificate is included and that any bid generated after~$\tend$ is rejected, delivering \emph{Simultaneous Release}. At settlement slot~$S$, only the winning bid is posted on-chain, delivering \emph{APE}.
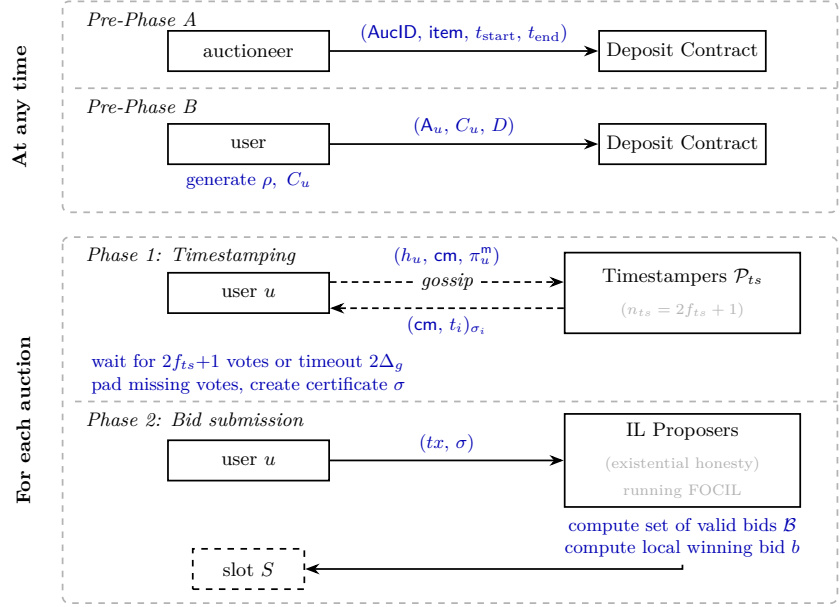
\begin{figure*}[t]
\centering
\resizebox{0.8\textwidth}{!}{
\begin{tikzpicture}[
  font=\small, >=Stealth,
  actor/.style={draw, thick, fill=white, rectangle,
    minimum width=2.6cm, minimum height=0.65cm, align=center},
  arr/.style={->, thick},
  dotarr/.style={->, thick, densely dashed},
  lbl/.style={font=\footnotesize, color=blue!70!black, align=center},
  phlbl/.style={font=\small\bfseries},
  sublbl/.style={font=\small\itshape, anchor=north west},
  divline/.style={dashed, gray!60, thick},
]

\draw[divline, rounded corners=3pt] (1.5, 0.55) rectangle (14.0, -2.85);
\node[phlbl, rotate=90, anchor=south] at (1.1, -1.15) {At any time};

\node[sublbl] at (1.75, 0.5) {\textit{Pre-Phase A}};
\node[actor] (aucA) at (4.5, -0.25) {auctioneer};
\node[actor] (dcA)  at (11.5, -0.25) {Deposit Contract};
\draw[arr] (aucA.east) --
  node[above, lbl]{$(\auctionid,\,\itemz,\,\tstart,\,\tend)$}
  (dcA.west);

\draw[divline] (1.7,-0.85) -- (14.0,-0.85);

\node[sublbl] at (1.75, -0.9) {\textit{Pre-Phase B}};
\node[actor] (usrB) at (4.5, -1.75) {user};
\node[actor] (dcB)  at (11.5, -1.75) {Deposit Contract};
\draw[arr] (usrB.east) --
  node[above, lbl]{$(\Addr_u,\, C_u,\, \depositamount)$}
  (dcB.west);
\node[lbl, anchor=north] at (4.5, -2.08)
  {generate $\secret$,\enspace$C_u$};

\draw[divline, rounded corners=3pt] (1.5, -3.25) rectangle (14.0, -9.15);
\node[phlbl, rotate=90, anchor=south] at (1.1, -6.2) {For each auction};

\node[sublbl] at (1.75, -3.3) {\textit{Phase 1: Timestamping}};

\node[actor] (usrTS) at (4.5, -4.15) {user $\user$};
\node[actor, minimum width=3.8cm, minimum height=1.3cm]
  (tsN) at (11.5, -4.15)
  {Timestampers $\timestampers$\\[3pt]{\scriptsize\color{gray!60}$(n_{ts}=2f_{ts}+1)$}};

\coordinate (upperL) at ([yshift=5pt]usrTS.east);
\coordinate (upperR) at ([yshift=5pt]tsN.west);
\path (upperL) -- (upperR) coordinate[midway] (upperMid);
\draw[densely dashed, thick] (upperL) -- ($(upperMid)-(0.55,0)$);
\node[font=\small\itshape, fill=white, inner sep=1.5pt] at (upperMid) {gossip};
\draw[dotarr] ($(upperMid)+(0.55,0)$) -- (upperR);
\node[lbl] at ($(upperL)!0.5!(upperR)+(0,0.42)$)
  {$(\handle,\,\commitment,\,\memberproof)$};

\draw[dotarr] ([yshift=-7pt]tsN.west) --
  node[below, lbl]{$(\commitment,\,t_i)_{\certificate_i}$}
  ([yshift=-7pt]usrTS.east);

\node[lbl, anchor=north, align=left] at (4.5, -5.0)
  {wait for $2f_{ts}{+}1$ votes or timeout $2\gossipdelay$\\[1pt]
   pad missing votes, create certificate $\certificate$};

\draw[divline] (1.7,-5.9) -- (14.0,-5.9);

\node[sublbl] at (1.75, -5.95) {\textit{Phase 2: Bid submission}};

\node[actor] (usrBid) at (4.5, -6.85) {user $\user$};
\node[actor, minimum width=3.8cm, minimum height=1.5cm]
  (ilP) at (11.5, -6.85)
  {IL Proposers\\[3pt]{\scriptsize\color{gray!60}(existential honesty)}\\[1pt]{\scriptsize\color{gray!60}running FOCIL}};
\draw[arr] (usrBid.east) --
  node[above, lbl]{$(\tx,\,\certificate)$}
  (ilP.west);

\node[lbl, align=center, anchor=north] (computeLbl) at (11.5, -7.65)
  {compute set of valid bids $\setbids$\\[1pt]compute local winning bid $\bid$};

\node[actor, dashed, minimum width=1.8cm]
  (slotS) at (4.5, -8.6) {slot $S$};
\draw[arr] (computeLbl.south) -- (computeLbl.south |- slotS.east) -- (slotS.east);

\end{tikzpicture}
}
\caption{Overview of the auction protocol. Pre-Phase~A registers an auction; Pre-Phase~B is a one-time user registration covering all later auctions. Phase~1 (timestamping) delivers \emph{Hiding} via a hidden commitment submitted to the timestamping oracle. Phase~2 (bid reveal) delivers \emph{Simultaneous Release} via a censorship-resistant inclusion scheme enforcing timestamp certificates received in Phase~1. At settlement slot~$S$, only the winning bid is posted on-chain, delivering \emph{APE}.}
\label{fig:overviewt}
\end{figure*}

\paragraph{Contributions.}
We summarize our contributions as follows.
\begin{enumerate}[itemindent=0pt,nosep]
  \item \emph{Property formalization.} We refine and formalize four properties for sealed-bid auctions in latency-sensitive blockchain settings: a strengthened hiding property capturing three leakage channels (bid value, bid existence, bidder identity); \emph{Simultaneous Release}, combining short-term censorship resistance with post-deadline exclusion, parameterized by two network-dependent gaps: the lead time an honest user must allow before the deadline, and the window in which a late adversarial bid is still admitted;
    \emph{No Free Bid Withdrawal}; and \emph{Auction Participation Efficiency}. Hiding and auction participation efficiency are posed as ideal targets together with attainable relaxations.

  \item \emph{Protocol design.} We design a protocol using two primitives -- a timestamping oracle and a censorship-resistant inclusion scheme -- which we instantiate using a $2\ftimestampers+1$ timestamping committee and a FOCIL-based inclusion list mechanism. Only the winning bid is settled on-chain; all other auction state lives off-chain.

  \item \emph{Implementation and benchmarks.} We implement the protocol's two zero-knowledge proofs -- the \emph{eligibility proof} $\memberproof$ proving anonymous deposit membership for the timestamping committee and the \emph{auction proof} $\idproof$ binding a handle to a specific auction ID for the IL committee -- using Groth16~\cite{DBLP:conf/eurocrypt/Groth16} over BN254 with Poseidon hashing~\cite{DBLP:conf/uss/0001KR0S21} in arkworks/Rust. The auction proof generates in $13$\,ms and verifies in under $1$\,ms; eligibility proofs for Merkle trees up to $2^{32}$ bidders generate in $47$--$159$\,ms and verify in ${\sim}1$\,ms.
\end{enumerate}

To the best of our knowledge, our protocol is the first to combine all four of these properties in a single design, by separating off-chain bid timestamping and inclusion enforcement from on-chain settlement. \Cref{tab:related-work} shows the comparison with other related schemes.

\section{Model and Notation}
\label{sec:model}

Our goal is to design a protocol that enables traditional sealed-bid auctions to be conducted on-chain while ensuring that every bid is included during the auction.
\subsection{Notation}
In pseudocode, the keyword \textbf{function} denotes a method of a smart contract, i.e., a method that is executed \emph{on-chain}, while keyword \textbf{offline} denotes a block of code executed by users \emph{off-chain}.
If a \textbf{require} statement is not satisfied, the contract reverts \emph{without} state changes.
When referring to zero-knowledge proofs, we denote by $\zkprove(\varStatement, \varWitness)$ a function that returns a proof $\pi$ that \varStatement is true with respect to \varWitness, and by $\zkverify(\varStatement, \pi)$ a function that verifies that $\pi$ is a valid proof for \varStatement.

\subsection{Network and Participants}
\label{subsec:network}

\paragraph{Network Model.} We assume a \emph{synchronous network} where any message sent by an honest party is guaranteed to be received by all other honest parties within a known maximum delay $\syncdelay$. This also means that we assume a $\syncdelay$-synchronized clock for all parties, i.e., the clocks they maintain are within $\syncdelay$ of each other. Synchrony is necessary for our setting: under asynchrony, an honest bid may never be delivered before the auction deadline, trivially breaking censorship resistance and rendering sealed-bid auctions impossible to implement with any meaningful timeliness guarantee.

\paragraph{Communication model.}
Parts of our protocol use an \emph{anonymous broadcast channel}. This channel allows any party to send a message such that (i) the message is delivered to all parties within delay $\gossipdelay$, and (ii) no party -- including a network-level adversary -- can link the message to its sender beyond the prior distribution over potential senders. Replies are routed through the same channel using a message-specific identifier (e.g., a freshly sampled tag included in the original message), without revealing the original sender's identity. Such channels can be instantiated using mix networks~\cite{chaum1981,nym}, anonymous broadcast protocols~\cite{dissent, riposte}, or Tor-based gossip overlays~\cite{tor} with appropriate latency-bandwidth
tradeoffs.

\paragraph{System Participants.} We call any participant in the protocol as a \emph{party}. Our system consists of three main types of parties:
\begin{itemize}[leftmargin=*,nosep]
    \item \textbf{Users ($\userset$):} A set of entities that wish to participate in an auction.
    A user $\user \in \userset$ creates a transaction $\tx$, which includes its contents (i.e., a bid) and a fee $\txfee$. A user uses the gossip network to interact with the timestamping replicas. To represent the value of a bid contained in the transaction $\tx$, we use the notation $\val{\tx}$

    \item \textbf{Timestamping Replicas ($\timestampers$):} A permissioned set of $\numtimestampers = 2 \ftimestampers + 1$  replicas, responsible for providing a timestamp for any transaction they receive. On receiving a transaction (or a blinded transaction, e.g., its hash), a timestamping replica $\timestamperc \in \timestampers$ signs the transaction after attaching the current local clock time. We assume an honest replica timestamps every transaction it receives, and that the adversary controls at most \ftimestampers timestamping replicas.

    \item \textbf{Inclusion List (IL) Proposers ($\ILproposers$):} A permissioned committee of \numILproposers parties, responsible for censorship resistance. Each inclusion list proposer $\ILproposer \in \ILproposers$ contributes to constructing the inclusion list by choosing transactions from the mempool. We assume the adversary controls at most \ILthreshold of the IL Proposers.
    The exact value for \ILthreshold depends on the inclusion-list mechanism used. In this paper, we will be using FOCIL~\cite{soispoke2024focil}, hence we assume existential honesty in the set $\ILproposers$, i.e., $\ILthreshold = |\ILproposers| - 1$.
\end{itemize}

In addition, an \textbf{Auctioneer} ($\auctioneer$) registers each auction by specifying its parameters in the on-chain deposit contract (see \Cref{subsec:prephaseA}). The auctioneer is not trusted for safety: any party can act as an auctioneer, and the protocol's security properties hold regardless of the auctioneer's behavior.

The sets $\timestampers$, $\ILproposers$, and $\userset$ can overlap, but for simplicity we assume in our analysis that the sets are disjoint.
For all of these sets, we call \emph{corrupted} a party controlled by the adversary -- in which case the party can behave arbitrarily -- and \emph{honest} otherwise.

\section{Desired Properties}
\label{sec:properties}

A secure sealed-bid auction protocol in our model must satisfy the following four core properties.
We argue that each is necessary in the sense that violating any one of them enables a concrete attack that undermines the auction:
\begin{itemize}
    \item dropping \emph{Hiding} allows the adversary to condition its bid on observed bids, resulting in violation of the sealed-bid part of the auction;
    \item dropping \emph{Simultaneous Release} allows an adversary to censor all rival bids and win at an arbitrarily low price or the ~\cite{fox2023censorship,sealedbid25};
    \item dropping \emph{\bindingProperty} allows a bidder to selectively abort after observing others' bids, reintroducing adverse selection via the ``free option'' problem~\cite{ens1};
    \item and dropping \emph{APE} transforms the auction into one with entry fees, changing the equilibrium bidding strategy and reducing participation~\cite{mcafee}.
\end{itemize}

\paragraph{Hiding.}
Existing definitions of hiding in sealed-bid auctions typically guarantee only value indistinguishability~\cite{sealedbid25,zeroauction}.
However, this is incomplete for latency-sensitive auctions.
First, the \emph{timing} of bid creation leaks information: in an auction where prices fluctuate, knowing when a bid was generated can reveal its value even if it is cryptographically hidden~\cite{sealedbid25}.
Second, the \emph{existence} of a bid leaks strategic information: if an adversary knows no other bids have been submitted, it can safely underbid.
Third, the \emph{identity} of the bidder can leak the expected bid amount based on the adversary's prior.
Thus, it is not enough to hide the value; one must also obfuscate which auction a bid belongs to and which user submitted it.

We capture these three channels in a single hiding definition.
Let $\textsc{AucSet}(\medianTimestamp)$ denote the set of auctions active at time $\medianTimestamp$, and let $\userset$ denote the set of registered users.

\begin{definition}[$\thiding$-Hiding]\label{def:hiding}
Let $\negl(\cdot)$ denote a negligible function.
For an auction with submission deadline~$\tend$, let $\thiding > \tend$ denote a point in time.
We define the following three games:

\begin{enumerate}[label=(\roman*),leftmargin=*,nosep]

    \item \textbf{Value Indistinguishability.}
    $\adversary$ outputs two transactions $(\tx_0,\tx_1)$.
    A challenger samples $b \xleftarrow{\$}\{0,1\}$ and submits $\tx_b$ to the protocol.
    At time $\thiding$, $\adversary$ outputs a guess $b'$.
    The advantage is
    $\mathsf{Adv}^{\mathsf{ind}}_{\adversary}(\lambda) := | \Pr[b'=b] - \tfrac{1}{2} |.$

    \item \textbf{Existential Obfuscation within $\textsc{AucSet}(\medianTimestamp)$.}
    $\adversary$ outputs a transaction $\tx$ at time $\medianTimestamp$ and two auction IDs $\auctionid_0, \auctionid_1 \in \textsc{AucSet}(\medianTimestamp)$.
    A challenger samples $b \xleftarrow{\$}\{0,1\}$ and submits $\tx$ for auction $\auctionid_b$.
    At time $\thiding$, $\adversary$ outputs a guess $b'$.
    The advantage is
    $\mathsf{Adv}^{\mathsf{exist}}_{\adversary}(\lambda) := | \Pr[b'=b] - \tfrac{1}{2} |.$

    \item \textbf{User Obfuscation within $\userset$.}
    $\adversary$ outputs a transaction $\tx$ and two user identifiers $\user_0, \user_1 \in \userset$.
    A challenger samples $b \xleftarrow{\$}\{0,1\}$ and submits $\tx$ on behalf of user $\user_b$.
    At time $\thiding$, $\adversary$ outputs a guess $b'$.
    The advantage is
    $\mathsf{Adv}^{\mathsf{user}}_{\adversary}(\lambda) := | \Pr[b'=b] - \tfrac{1}{2} |.$
\end{enumerate}

Define $\mathsf{Adv}^{\mathsf{hiding}}_{\adversary}(\lambda) := \max\{ \mathsf{Adv}^{\mathsf{ind}}, \mathsf{Adv}^{\mathsf{exist}}, \mathsf{Adv}^{\mathsf{user}} \}$.
A protocol satisfies $\thiding$-\emph{Hiding} if for every PPT $\adversary$: $\mathsf{Adv}^{\mathsf{hiding}}_{\adversary}(\lambda) \leq \negl(\lambda)$.
\end{definition}

Games (ii) and (iii) are parameterized by the sets $\textsc{AucSet}$ and $\userset$, which bound the adversary's uncertainty. Larger sets yield stronger guarantees.

\begin{remark}[Ideal Hiding]\label{rem:ideal-hiding}
An ideal hiding property would replace \emph{Existential Obfuscation} with full \emph{Existential Indistinguishability} (the adversary cannot tell whether \emph{any} bid was submitted) and \emph{User Obfuscation} with full \emph{User Anonymity} (over all parties, not just registered users).
We leave the question of whether a protocol can achieve this ideal as an open problem.
\end{remark}

\paragraph{Simultaneous Release.}
This property combines censorship resistance and post-auction exclusion. Fox et al.~\cite{fox2023censorship} show the importance of censorship resistance for on-chain auctions -- without it, a malicious bidder can always ensure that honest bidders have lower utility than in classical auctions. Post-auction exclusion states that once the auction completes and bid values may lose privacy, no adversarial bids can be generated, preventing the adversary from inserting bids after that point.

In blockchains, censorship resistance has traditionally been defined as a form of liveness: an honest user's transaction is guaranteed to be included eventually, assuming blocks continue to be produced. For sealed-bid auctions, eventual inclusion is insufficient. Bids must be included \emph{before} the auction deadline~$\tend$ for the bidder to participate. Even a short delay, e.g., a malicious proposer withholding an honest bid and instead including its own lower bid just before the deadline, can prevent the honest bid from being included in the auction.

While prior work has introduced short-term notions of censorship resistance, these guarantees apply only \emph{conditional on the block for that slot being produced}~\cite{sealedbid25}. If the adversary controls the final proposer before the auction deadline~$\tend$ and knows that no subsequent block will be built, it can prevent timely honest bids from being included simply by not producing the final block, while still including its own bid in the preceding block. Auctions therefore require a form of short-term censorship resistance that guarantees inclusion in the auction's input set before~$\tend$, independent of whether a particular block is produced. We refer to this as \emph{Short-Term (ST) Censorship Resistance}. 

In a synchronous network with zero gossip delay and perfectly synchronized clocks, ST Censorship Resistance and Post-Auction Exclusion collapse into a single hard deadline at $\tend$. In practice, however, honest users submit bids through a gossip network, incurring a propagation delay of up to $\gossipdelay$ in addition to clock skew $\syncdelay$, whereas an adversary can submit directly and may submit up to $\syncdelay$ after the deadline. We make this asymmetry explicit by parameterizing the property by two values: $\delta_i$, the lead time an honest user must allow before $\tend$ to be guaranteed inclusion, and $\delta_e$, the late-arrival window in which an adversarial bid may still be admitted.

\begin{definition}[$(\delta_i, \delta_e, \tend)$-Simultaneous Release]
\label{def:simulrelease}
A sealed-bid auction protocol with deadline $\tend$ satisfies \emph{$(\delta_i, \delta_e, \tend)$-Simultaneous Release} if both of the following hold:
\begin{itemize}
    \item \textbf{$\delta_i$-ST Censorship Resistance.} For every auction and every honest bidder $\user$, if $\user$'s bid~$\tx$ is sent to the protocol at least $\delta_i$ before $\tend$, then $\tx$ is included in the auction's input bid set used to determine the outcome.
    \item \textbf{$\delta_e$-Post-Auction Exclusion.} For every auction, if an adversary's bid~$\tx$ is sent to the protocol more than $\delta_e$ after $\tend$, then $\tx$ is not included in the auction's input bid set.
\end{itemize}
\end{definition}

Setting $\delta_i = \delta_e = 0$ recovers the ideal one-shot deadline; the achievable values depend on the network. Concretely, $\delta_i$ must absorb the time an honest bid spends traversing the gossip overlay and obtaining its timestamp certificate, while $\delta_e$ must absorb the clock skew an adversary can exploit at the deadline. Our construction achieves $(\gossipdelay+\syncdelay,\,\syncdelay,\,\tend)$-Simultaneous Release (\Cref{thm:Simul-Release}).

\paragraph{\bindingProperty.}
From the \emph{Hiding} property we have that, until time \thiding, the adversary has no information about the bids that have been submitted. After time \thiding, the protocol may allow bids to be revealed.
Of course, for the \emph{Hiding} property to make sense, the protocol will not allow new bids to be created after this point, ensured by \emph{Simultaneous Release}. However, this is not sufficient -- if the adversary has the ability to \emph{cancel} bids, i.e., to hide their existence or remove them from the protocol, then this would be effectively the same as being able to create bids after \thiding. For example, an adversary could submit many bids and suppress all but the one that yields the most favorable outcome.
Such manipulations arise broadly in commit-and-reveal schemes, where the ability to decide whether a bid becomes visible after observing partial information from others enables strategic withholding or selective abort. Hence, we require a \emph{\bindingProperty} property, complementing the \emph{Hiding} property.

\begin{definition}[\bindingProperty]\label{def:binding}
A sealed-bid auction protocol satisfies the \emph{\bindingProperty} property if the following holds. Consider an execution of the protocol in which all users are honest after time~$\thiding$ and the winning bid is the bid contained in transaction~$\tx$. Now consider a second execution that allows adversarial users  after time~$\thiding$, resulting in a winning bid contained in transaction~$\tx'$. If the bid~$\bid'$ in $\tx'$ is strictly smaller than the bid~$\bid$ in $\tx$, i.e.,  $\bid' < \bid$,  then the user submitting $\tx$ must incur a cost for doing so.

\end{definition}

The intuition is that, if an execution~$e$ exists where transaction $\tx'$ contains the winner bid $\bid'$ and $\bid' < \bid$, whereas the honest execution returns $\bid$ as the winner bid, then the user that submitted bid $\bid$ has effectively withdrawn \bid in execution $e$. In other words, the adversary is not able to change the winner bid after time \thiding to a lower bid $\bid'$ without incurring a cost.

\paragraph{Auction Participation Efficiency.}
We introduce a fourth property -- \emph{Auction Participation Efficiency} -- which captures how closely an on-chain auction emulates the payoff dynamics of its classical counterparts.
In traditional auctions, only the winning bidder(s) pay their bid amounts, while all other participants incur no cost for participation.
In contrast, on-chain implementations require every bidder to submit a transaction containing their bid, each of which must pay an inclusion fee regardless of the auction outcome. This leads to a situation in which most parties end up with a negative utility at the end of the auction; their expected value is inflated by the cost incurred when they do not win. 
This changes the auction to a design with entry fee~\cite{mcafee,krishna2009auction}, which changes the strategy that a bidder must take. 

\begin{definition}[\ape-Auction Participation Efficiency]\label{def:ape}
An on-chain auction protocol satisfies \emph{Auction Participation Efficiency (APE)} if there is at most $\ape$ cost incurred by losing bidders.
\end{definition}

In an ideal world, $\ape = 0$ fully recreates the payoff structure of offline auctions.
However, achieving $\ape = 0$ together with \emph{Hiding} and \emph{\bindingProperty} is impossible: if participating is entirely free, a bidder can submit a hidden bid, observe the outcome after reveal, and selectively withhold its commitment whenever doing so is profitable -- precisely the free-option attack that \emph{\bindingProperty} is designed to prevent.
Hence any protocol satisfying all four properties must impose a non-zero cost on bidders, which may be amortized across multiple auctions.
Our protocol achieves this by requiring a one-time deposit that is returned upon honest participation, so that the amortized cost per auction tends to~$0$.

\section{Background}
\paragraph{Overview of the FOCIL Protocol.}\label{sec:focil-overview}
FOCIL (Fork-choice enforced Inclusion Lists)~\cite{eip-7805} has been proposed as a mechanism to improve transaction-inclusion guarantees on Ethereum by imposing constraints on Ethereum block builders.
An instance of the FOCIL protocol is run for each Ethereum slot, which, according to the Ethereum consensus protocol, has a duration of 12 seconds. Each instance involves a set of \emph{IL Proposers}, a set of \emph{Attesters}, and a set of \emph{Validators}, one of which is the \emph{block builder}. An instance works as follows.

A set of IL proposers is selected from the validator set. Each IL proposer monitors the public mempool and constructs an IL of Ethereum transactions. Each IL can be at most $8$\,KB, and proposers are free to choose their own construction policy -- for example, prioritizing transactions by highest priority fee. Once built, the IL is gossiped to validators and to the block builder for the current slot.
Once a validator receives an IL, it verifies certain conditions: (1) the size of the IL does not exceed the maximum allowed size, (2) the IL is for the current slot, (3) it is signed by an IL proposer of the current slot, and (4) that IL proposer has not sent a different IL for the current slot. If verified, the validator propagates the IL.

The block proposer collects ILs and builds a block consisting of the transactions in \emph{all} valid ILs it has received -- except for transactions that do not fit in the block, which the block proposer may choose to omit. The block is propagated to the rest of the validators.
A set of IL attesters is also chosen from the set of all validators. An attester receives the block for the current slot and only votes for it if it does not omit any transaction from an IL locally seen -- except for transactions that do not fit in the block. A block not signed by sufficiently many attesters is not considered canonical and cannot be part of the canonical chain.

Looking forward, we will use FOCIL as the second sub-protocol in our construction, but with some modifications that we introduce later.

\paragraph{Blockchain Protocol.}\label{sec:pos-overview}
We assume a standard single-leader proof-of-stake (PoS) blockchain with a deterministic slot schedule that supports smart contracts (e.g., Ethereum). The deterministic slot schedule means that slots occur at precise, pre-defined intervals known to all participants, allowing the network to remain synchronized. In each slot a validator acts as the \emph{leader}, collects transactions, executes them, and propose a new block, while other validators attest to finalize it under the protocol's fork-choice rule. This fork-choice incorporates validity conditions, for our construction specified in \Cref{sec:phase2}.

This base protocol ensures consensus on both the order of transactions and the resulting on-chain state. In our construction, we only modify it to require that the inclusion of winning auction bids is enforced through a fork-choice–based mechanism, as described in \Cref{sec:phase2}.

\section{Protocol Description}\label{sec:protocol}
Before presenting the details, we describe the protocol as an abstraction from the point of view of the parties that interact with it, the users $\userset$ and the auctioneer \auctioneer. The former wish to submit a bid to an auction. Each user registers in a pre-phase by locking a deposit and generating a private commitment that will later be used to prove eligibility. The auction protocol will require a user $\user \in \userset$ to act twice, i.e., to send a message, receive responses from the protocol, and then send a second message during the bidding phase. The auctioneer is specified in the underlying PoS  protocol. It registers the auction during a preparatory phase, specifying parameters such as the auction identifier, start and end times, and the item being auctioned. The auctioneer then obtains the bids from our protocol.

Our protocol enables censorship-resistant sealed-bid auctions. It combines two sub-protocols: a timestamping network that provides verifiable bid creation times, and a censorship-resistant inclusion mechanism based on ILs. Together, they achieve the properties presented in \Cref{sec:properties}. At a high level, the protocol operates as follows.

\begin{description}[leftmargin=0pt, labelindent=0pt]

\item\textbf{Pre-Phase A (Auction Registration):} The auctioneer registers the auction $\auctionid$ in the global deposit contract~$\depositcontract$.  During this process, the auctioneer specifies the parameters and reserves settlement capacity on-chain. This ensures unique and verifiable auctions and limits concurrent auctions to guarantee sufficient on-chain settlement capacity.

\item\textbf{Pre-Phase B (User Deposit and Membership Proof):}
Each user $\user$ registers by locking a deposit $D$ in the  global contract~$\depositcontract$ and committing to a private secret $\secret_\user$,  forming $C_\user = \textsc{Hash}(\secret_\user)$ that is stored in a public Merkle-tree registry.
The deposit mechanism will be used to ensure that bidders cannot submit multiple bids and selectively reveal only the most favorable one. Each deposit will remain locked until the user, for each auction it has participated in, either reveals its bid or proves that the bid was smaller than the winning bid.
This will assist in achieving the \emph{\bindingProperty} property.

We remark that a user~$\user$ registers only once with a fixed deposit~$\depositamount$. The deposit remains active until the user withdraws it and enables $\user$ to participate in all auctions under the global contract~$\depositcontract$, as long as it follows the protocol in each auction. 

\item\textbf{Phase 1 (Timestamping):}
For every auction $\auctionid$, the user proves its eligibility to participate by deriving a pseudonymous handle $\handle = \textsc{Hash}(\secret_\user \parallel \auctionid)$ and generating a zero-knowledge \emph{proof of membership}, showing that $\handle$ corresponds (through the secret $\secret_\user$) to an active deposit.
An eligible user obtains a timestamp for a commitment of its bid by interacting with the timestampers $\timestampers$.
A timestamper only signs a commitment if the user provides a handle \handle and a verifying proof.
Collectively, \timestampers certify the creation time of the bid without learning its content, producing a \emph{timestamp certificate} as public evidence of bid existence prior to the auction deadline.
This phase will assist in achieving the \emph{$\thiding$-Hiding} property, ensuring that all bids can be revealed simultaneously at release time while remaining hidden beforehand.

\item\textbf{Phase 2 (Bid Submission and Inclusion):}
Each user now reveals its bid by submitting it, together with the timestamp certificate, to the inclusion list proposers $\ILproposers$, which ensure that all valid bids reaching the network before the auction deadline are eligible for inclusion.
The protocol enforces that the winning bid is included by requiring validators to check an \emph{external validity predicate} -- namely, that the block contains the highest valid bid among all declared inclusion lists. In our Ethereum-based instantiation, this predicate is enforced through the fork-choice rule (see \Cref{sec:phase2}), but the construction generalizes to any consensus mechanism that supports such a predicate.
This phase will assist in achieving the \emph{ST censorship resistance} property.
\end{description}

\subsection{Pre-Phase A: Auction Registration}
\label{subsec:prephaseA}

Before any user deposits or bids can occur, each auctioneer must register its auction in the global on-chain deposit contract.
This pre-phase ensures that all subsequent auctions are uniquely identifiable, at most $\numauctions$-overlapping, i.e., our protocol supports $\numauctions$ parallel auctions.

Each auctioneer~$\auctioneer$ registers an auction in a global contract~$\depositcontract$ that maintains the registry of deposits and auction parameters.
In \Cref{alg:pre-phases} we show the algorithm exposed by \depositcontract for this purpose.
Each registered auction is identified by a unique and strictly increasing \emph{Auction ID} (\auctionid).
The protocol supports at most $\numauctions$ auctions
ongoing at the same time.
For each auction, the auctioneer specifies $
(\auctionid, \itemz, \tstart, \tend),
$
where $\itemz$ denotes the asset being auctioned, and $\tstart < \tend$ define the auction's active window.

We assume a proof-of-stake blockchain protocol with a deterministic slot schedule.
A small fraction of blockspace in each slot is reserved for auction-settlement transactions.
This reserved blockspace unconditionally includes the winning bids on chain; since only a single winning bid must be posted per auction (see \Cref{sec:phase2}) and at most \numauctions auctions run in parallel, the total reserved capacity required per slot is small.
If one or more preceding slots fail to produce a block, the protocol automatically increases the reserved capacity in the next available slot to ensure that all pending auction settlements can be included.

\newcommand{\rcomment}[1]{\unskip\hfill{\scriptsize$\triangleright$~#1}}

\begin{algorithm}[htb!]
  \caption{Pre-Phases A and B: Auctioneer Setup and User Deposit}
  \label{alg:pre-phases}
  \footnotesize
  \begin{algorithmic}[1]
    \Require Auctioneer identifier $\auctioneer$, user \user
    \Ensure Register auctions (at most $\numauctions$-overlap) and users with deposits
  \end{algorithmic}
  \algrule
  \noindent\begin{minipage}[t]{0.49\textwidth}
  \footnotesize
  \begin{algorithmic}[1]
    \Statex \textbf{Pre-Phase A: Auctioneer Registration}
    \Statex \textbf{On-chain functions on contract \depositcontract:}
    \State \textbf{global} $\varCounter \gets 0$
    \State \textbf{global} $\varPerSlot[i] \gets 0, \forall i \in \CN$
    \State \textbf{global} $\varAllAuctions \gets \emptyset$ \rcomment{map, holds auction param.}
    \vspace{3pt}
    \Function{regAuction}{$\auctioneer, \varItem, \tstart, \tend$}
      \For{$t \in [\tstart, \tend]$} \rcomment{max $\numauctions$ per slot}
        \State $\textbf{require } \varPerSlot[t] < \numauctions$
        \State $\varPerSlot[t] \gets \varPerSlot[t] + 1$
      \EndFor
      \State $\auctionid \gets \varCounter$
      \State $\varAllAuctions[\auctionid] \gets (\varItem, \tstart, \tend, S)$
      \State $\varCounter \gets \varCounter + 1$
    \EndFunction
    \vspace{3pt}
    \Function{getAucRecord}{$\auctionid$}:
      \State \textbf{return} $\varAllAuctions[\auctionid]$
    \EndFunction
    \algrule
    \Statex \textbf{Pre-Phase B: User Deposit}
    \Statex \textbf{Local for user $\user$:}
    \Offline{\deposit}{\depositamount} \rcomment{Run by user with address $\Addr_\user$}
      \State $\secret_\user \gets_r \{0,1\}^\lambda$
      \State $C_\user \gets \textsc{Hash}(\secret_\user)$
      \State \registerDeposit($\Addr_\user, C_\user, \depositamount$)
    \EndOffline
    \vspace{3pt}
    \Statex \textbf{On-chain functions on contract \depositcontract:}
    \State \textbf{global} \varDeposits \rcomment{Merkle tree, holds user deposits}
    \State \textbf{global} \varInactives \rcomment{Users made inactive}
    \State \textbf{global} \varSlashed \rcomment{Slashed deposit handles}
    \State \textbf{global} \varMinDeposit \rcomment{Minimum deposit}
    \vspace{3pt}
    \Function{\registerDeposit}{$\Addr_\user, C_\user, \depositamount$}
      \State \textbf{require} sender $= \Addr_\user$ \rcomment{caller is $\Addr_\user$}
      \State \textbf{require} $\depositamount \geq$ minDeposit
      \State record $(\Addr_\user, C_\user, \depositamount)$
      \State insert $(\Addr_\user, C_\user)$ in \varDeposits
    \EndFunction
  \end{algorithmic}
  \end{minipage}\hfill
  \begin{minipage}[t]{0.49\textwidth}
  \footnotesize
  \begin{algorithmic}[1]
    \setcounter{ALG@line}{26}
    \Statex \textit{Called by IL Proposer if $|\useractivity[\user]| \geq \numauctions$ ($\useractivity$ defined in \Cref{alg:bidreveal-winner}):}
    \Function{\makeOffline}{$\Addr_\user$}
      \State \textbf{require} sender $\in \ILproposers$ \rcomment{IL Proposers only}
      \State \textbf{require} ($\Addr_\user, C_\user$) in $\varDeposits$
      \State start $\numauctions$ auction timer
      \State at end of $\numauctions$ timer:
      \State \quad insert ($\Addr_\user, C_\user$) in $\varInactives$
      \State \quad remove ($\Addr_\user, C_\user$) from $\varDeposits$
    \EndFunction
    
    \vspace{3pt}
    \Function{updateActivity}{$\Addr_\user, \Pi$}
      \State \textbf{require} sender $\in \ILproposers \lor $ sender $=\Addr_\user$
      \State \textbf{require} $\varInactives$ contains ($\Addr_\user, C_\user$)
      \State \textbf{call} \updateHistory($\Pi$)
      \State insert ($\Addr_\user, C_\user$) in $\varDeposits$
      \State remove ($\Addr_\user, C_\user$) from $\varInactives$
    \EndFunction
    
    \vspace{3pt}
    \Function{\slashUsr}{$\handle,(\commitment,\memberproof),(\commitment',{\memberproof}^{'})$}
      \State \textbf{require} $\verifyEligibility(\handle,\commitment,\memberproof)$
      \State \textbf{require} $\verifyEligibility(\handle,\commitment',{\memberproof}^{'})$
      \State \variable{slashed}.add($\handle$)
    \EndFunction
    
    \vspace{3pt}
    \Function{slash}{$\Addr_\user, \handle,\Pi$}
      \State verify signature on $\Pi$ from $\Addr_\user$
      \If{$\handle \in \Pi \text{ and }\handle \in \varSlashed$}
      \State remove $(\Addr_\user, C_\user)$ from \varDeposits
      \EndIf
    \EndFunction
    
    \vspace{3pt}
    \Function{\withdraw}{$\Addr_\user,\Pi$}
      \State verify message signed by user
      \State verify signature on $\Pi$ from $\Addr_\user$
      \State \textsc{updateHistory}($\Pi$)
      \State remove $(\Addr_\user, C_\user)$ from \varDeposits
      \State pay $\varMinDeposit$ to $\Addr_\user$
    \EndFunction
  \end{algorithmic}
  \end{minipage}
\end{algorithm}

Note that with Pre-Phase~A, we ensure the uniqueness and public verifiability of auction results, as all participants agree on the existence and the parameters of each auction. We also ensure that the number of concurrent auctions remains limited, so that sufficient capacity exists to unconditionally enforce the inclusion of bids. This foundation enables subsequent phases to rely on consistent auction timelines and prevents equivocation by auctioneers.
In practice, to prevent the auction contract from being dominated by low-value auctions, one could introduce a permissioned set of auctioneers, a reputation-based mechanism, or a dynamic registration fee for creating auctions. We discuss these design choices further in~\Cref{sec:discussion}.

\subsection{Pre-Phase B: Deposit}
\label{subsec:prephaseB}

We now describe how a user~$\user \in \userset$ \emph{registers} for participation in our auction mechanism, in a way that will later allow \user to prove eligibility to participate in an auction. The protocol is shown in \Cref{alg:pre-phases}.

\paragraph{Merkle-Tree-Based Registry.}
As shown in function $\deposit()$ of \Cref{alg:pre-phases}, user $\user$ deposits a deposit amount $\depositamount$ by committing to a private secret~$\secret_\user$, sampled uniformly from $\{0,1\}^\lambda$ and setting:
\[
C_\user = \textsc{Hash}(\secret_\user).
\]
The contract \depositcontract records $(C_\user, \depositamount)$ in a Merkle-tree-based registry. For a slot $\slot$, an observer (user or timestamper) can determine the root of the Merkle tree \rootslot, as well as all registered users and their commitments, $(\Addr_\user, C_\user)$.
This will later serve as public reference for zero-knowledge membership proofs -- users will use \rootslot to create the proof and timestampers will use the same \rootslot to verify the proof.

\paragraph{Activity requirement.}
Each deposit must remain active by submitting either 
    (1) a valid bid commitment, such that the bid is revealed to each timestamper in the reveal phase (winning or losing), or
    (2) a no-operation proof of non-participation,
in every auction until withdrawn. The proofs can be submitted after the auction ends, but being inactive in more than $\numauctions$ auctions  can trigger a $\makeOffline()$ call by an IL Proposer on-chain. Once the user becomes inactive, it must submit a proof of non-participation to regain eligibility. Failure to do so invalidates the deposit, stopping the user from withdrawing the amount. If such proof does not exist (due to misbehaviour), the bidder can never withdraw the deposit, and is effectively slashed.

\paragraph{Withdrawal.}
A withdrawal can be initiated by referencing the deposit record and signing with the key associated with the address $\Addr_\user$ of user \user.
At that time, the deposit of \user must be active, according to the activity requirement conditions,
and \user must have revealed its bids in all auctions in which it has participated.
If \user has not participated in some concluded auctions (these can be at most $\numauctions$, otherwise $\makeOffline()$ would have been called, stopping withdrawals), it can submit a proof of non-participation using $\proveInactivity()$ (\Cref{alg:bidreveal-winner}) for these and for the currently ongoing $\numauctions$ auctions, verifying that it is not currently using its deposit.
The entry for the bidder is instantly removed from the list of bidders, and \user is refunded its deposit.
We show this in function $\withdraw()$.

\paragraph{Inactivity.}
If \user has not revealed its bids (winning or losing; see \Cref{sec:phase2}) for $\numauctions$ auctions, then any IL Proposer can call the function $\makeOffline()$ to start an $\numauctions$-long auction window, after which \user will be removed from the list of bidders and added to an inactive-bidders list.

During this period, \user can call function $\updateActivity()$ to directly submit a proof of non-participation for the auctions it did not participate in, or reveal a bid lower than the winning bid for that auction on-chain (we explain in \Cref{sec:phase2} how the user generates the proof of non-participation $\Pi$ using function $\proveInactivity()$). Note that during the $\numauctions$-auction window, \user can still get its bid timestamped, but if \user is inactive when the auction ends then the bid can never be revealed as the winning bid.

According to the above, users need to submit proof of non-participation or reveal bids for a certain number of past auctions, at most $3\numauctions-1$ ($\numauctions$ auctions before the window starts, $\numauctions$ auctions during the window, and $\numauctions-1$ other auctions that might have started before the end of auctions in the $\numauctions$ auction window). This is re-stated in~\Cref{lem:max-inactivity}.

\paragraph{Slashing users.}
A user can be slashed if a handle $\handle$ and two different bid commitments $\commitment, {\commitment}'$ are given as evidence to function $\slashUsr()$, along with verifying proofs $\memberproof, {\memberproof}'$. This evidence can be provided by timestamperes (see function $\signCom()$ in \Cref{alg:timestamp}). However, the true identity of the user is still unknown at this point, and thus $\handle$ is marked as \varSlashed. Whenever the user tries to use that handle (either to prove \varNOP or reveal its bid), \textsc{slash} is called on the user's address, with requirement that one of the handles that the user sent is in the slashed set.

\paragraph{Security Rationale.}
This pre-phase will later assist in achieving the following properties:
\begin{enumerate*}
\item \emph{Duplicate Prevention and Sybil Resistance}: Each bid handle $\handle$ is unique per auction and per deposit, thus the same deposit cannot be used for multiple bids.
\item \emph{Anonymity}: The ZK membership proofs hide $\Addr_\user$ and $\secret_\user$ while proving that \user has a deposit that is active for the duration.
\item \emph{\bindingProperty}: Users must either bid or explicitly issue a \varNOP proof each round, and no user cannot have more than $3\numauctions-1$ auctions (from~\Cref{lem:max-inactivity}) for which such a proof does not exist.
\item \emph{Withdrawal Safety}: The enforced waiting period ensures no pending bids remain, hence the user can safely withdraw its deposit.
\end{enumerate*}

\subsection{Phase 1: Timestamping}\label{sec:timestamping}
\paragraph{Goal.}
Produce a public, verifiable certificate that a bid commitment existed \emph{before} the auction deadline, without revealing the bid or the bidder's identity.

\paragraph{Inputs and outputs.}
A user $\user$ with address $\Addr_\user$ who registered a deposit in \Cref{alg:pre-phases} with secret $\secret_\user$ wishes to submit a bid transaction $\tx$ for auction $\auctionid$.
The output of the phase is a \emph{timestamp certificate} $\certificate \;=\; \{(\commitment,\clocktime_i)_{\sigma_i} : i \in \timestampers\}$,
together with its \emph{median timestamp} $\medianTimestamp=\mathrm{median}(\{\clocktime_i : i\in \timestampers\})$. In this output, at most $\ftimestampers$ timestamps can be from a Byzantine timestamper, thus at least $\ftimestampers+1$ timestamps exist from honest timestampers. If a timestamp is not received from a timestamper, then it is considered to be $\infty$.

\paragraph{Zero-Knowledge proof of membership.}
User \user must first prove that it has registered a deposit in the registry with root \rootslot. We show this in function $\proveEligibility()$ in \Cref{alg:timestamp}.
Specifically, \user creates a pseudonymous \emph{handle} $\handle$ and a commitment \commitment to its transaction \tx (which contains a bid):$ \handle = \textsc{Hash}(\secret_\user \parallel \auctionid), \quad
  \commitment = \textsc{Hash}(\tx \parallel C_\user)$.
Observe that the handle is unique per auction and deposit, hence the user cannot submit more than one bid for each deposit made in the deposit contract, and that the commitment blinds the actual bid of the user.

Furthermore, $\user$ creates a zero-knowledge proof~$\memberproof$
with public elements (statement) \rootslot, \handle, and \commitment,
and private elements (witness) $\secret_\user$, $\Addr_\user$, the Merkle path from \rootslot to $(\Addr_\user, C_\user)$, \auctionid, and $\tx$,
attesting that
\begin{align*}
    &\handle = \textsc{Hash}(\secret_\user \parallel \auctionid)
    \;\; \land\;\; \exists\; C_\user : \\
      &
     \big(C_\user = \textsc{Hash}(\secret_\user)
      \land(\Addr_\user, C_\user) \in \textsc{merkle}(\rootslot)\;\;\land\; \commitment = \textsc{Hash}(\tx \parallel C_\user) \big)
\end{align*}
Essentially this is a proof of membership in \rootslot. Specifically, the proof demonstrates that:
(i) the user knows a secret $\secret_\user$ which produces the public handle $\handle$,
(ii) the secret $\secret_\user$ corresponds (through the relationship $C_\user = \textsc{Hash}(\secret_\user)$) to a deposit record $(\Addr_\user, C_\user)$ in the registry with root $\rootslot$,
(iii) the commitment \commitment has been created by hashing some transaction \tx and $C_\user$, and
(iv) the deposit is active.
At the same time, no information about $\Addr_\user$ or $\secret_\user$ is revealed.
Note that \auctionid is part of the handle, but \memberproof does not prove anything about it -- we will use this later in \Cref{sec:phase2}.

\paragraph{Protocol for users to obtain timestamp certificate.}
The user sends the handle \handle, bid commitment \commitment, and eligibility proof \memberproof to the timestampers \timestampers through the gossip network, as shown in function $\submitCom()$.
User \user waits to collect $2\ftimestampers+1$ valid responses from parties in \timestampers, and for at most $2\gossipdelay$ time. After that timeout, \user assumes all other parties have sent $\infty$ as their timestamp.
When it receives $2\ftimestampers+1$ responses, \user runs $\receiveTimestamp()$ to create the timestamp certificate \certificate with median timestamp \medianTimestamp, thereby concluding this phase of the protocol.

\paragraph{Protocol for timestampers.}
Upon receiving such request from a user, the timestampers run the function $\signCom()$ of \Cref{alg:timestamp}, which, if successful, signs the commitment and gossips the timestamp certificate to the user.
In order to verify the proof, the timestampers locally run function $\verifyEligibility()$, which contains the verification part of the proof constructed in $\proveEligibility()$. It recovers the Merkle root \rootslot from the slot $\slot$ and verifies that the proof $\memberproof$ is valid for the handle $\handle$ and commitment $\commitment$.

\begin{algorithm}[htb!]
  \caption{Phase 1: ObtainTimestamp}
  \label{alg:timestamp}
  \footnotesize
  \begin{algorithmic}[1]
    \Require User \user (with address $\Addr_{\user}$ and deposit secret $\secret_\user$), transaction \tx, timestamping validators \timestampers.
    \Ensure Obtain certificate \certificate and median timestamp \medianTimestamp.
  \end{algorithmic}
  \algrule
  \noindent\begin{minipage}[t]{0.50\textwidth}
  \footnotesize
  \begin{algorithmic}[1]
    \Statex \textbf{Local for user $\user$:}
    \Offline{\proveEligibility}{$\secret_\user, \Addr_\user, \auctionid, \tx, \slot$}:
      \State $C_\user \gets \textsc{Hash}(\secret_\user) ;\; \handle \gets \textsc{Hash}(\secret_\user \parallel \auctionid)$
      \State $\commitment \gets \textsc{Hash}(\tx \parallel C_\user)$
      \State $\rootslot \gets $ Merkle root for \varDeposits at slot $\slot$
      \State $\varPath \gets $ path from \rootslot to $(\Addr_\user, C_\user)$
      \State $\memberproof \gets$ \zkprove$((\rootslot, \handle, \commitment),(\secret_\user,\Addr_\user,\varPath, \tx))$
      \State \Return $(\handle, \commitment, \memberproof)$
    \EndOffline
    
    \vspace{3pt}
    \Offline{\submitCom}{$\tx, \secret_\user, \Addr_\user, \auctionid, \slot$}:
      \State $(\handle, \commitment, \memberproof) \gets$ \proveEligibility($\secret_\user, \Addr_\user, \auctionid, \tx, \slot$)
      \State $(\handle, \commitment, \memberproof,\slot) \xrightarrow{\textsc{gossip}} \timestampers$
    \EndOffline
    
    \vspace{3pt}
    \Statex \textit{On receiving $2\ftimestampers+1$ responses or waiting for $2\gossipdelay$ time from \submitCom, and padding the timestamps not received with $\infty$:}
    \Offline{\receiveTimestamp}{$\{(\commitment, \clocktime_j)\}_{j \in[2\ftimestampers+1]}$}
      \State $\medianTimestamp \gets \mathrm{median}(\{\clocktime_1,\dots,\clocktime_{2\ftimestampers+1}\})$
      \State $\sigma \gets \{(\commitment,\clocktime_j)_{\sigma_j} : j\in [2\ftimestampers+1]\}$.
      \State \Return $(\tau, \sigma)$.
    \EndOffline
  \end{algorithmic}
  \end{minipage}\hfill
  \begin{minipage}[t]{0.48\textwidth}
  \footnotesize
  \begin{algorithmic}[1]
    \setcounter{ALG@line}{16}
    \Statex \textbf{Local for timestamper $\timestamper{i}$:}
    \Statex \textbf{local} ids $\gets \{\}$ \rcomment{Handles observed by \timestamper{i}}
    \Offline \verifyEligibility{$\handle, \commitment, \memberproof, \slot$}:
      \State $\rootslot \gets $ Merkle root for \varDeposits at slot $\slot$
      \State \textbf{require} $\zkverify((\rootslot, \handle, \commitment), \memberproof)$
    \EndOffline
    
    \vspace{3pt}
    \Statex \textit{Run by $\timestamper{i}$ on receiving $(\handle, \commitment, \memberproof,\slot)$:}
    \Offline{\signCom}{$\commitment, \handle, \memberproof,\slot$}
      \State $\clocktime_i \gets$ local clock time
      \State \textbf{require} $\verifyEligibility(\handle, \commitment, \memberproof, \slot)$
      \If{ids$[\handle] = (\commitment'\neq\commitment,{\memberproof}^{'},\slot')$}
            \State \textsc{slashUsr}$(\handle,(\commitment,\memberproof,\slot),(\commitment',{\memberproof}^{'},\slot'))$
      \Else
      \State ids.add$(\handle, \commitment, \memberproof, \slot)$
      \State $(\commitment,\clocktime_i)_{\sigma_i} \xrightarrow{\textsc{gossip}} \user$ \rcomment{$\sigma_i$ is the sig. of $\timestamper{i}$}
      \EndIf
    \EndOffline
  \end{algorithmic}
  \end{minipage}
\end{algorithm}

\paragraph{Security intuition.}
Because at most $\ftimestampers$ replicas are faulty, the $(\ftimestampers{+}1)$-st order statistic (the median) is within the interval spanned by honest clocks (\Cref{lem:median-robustness}).
This bounds the median value for all users to be between the timestamps provided by two honest parties and prevents early or late skew beyond a threshold.
Since the tuple contains only~$\commitment$, a handle $\handle$ and is sent back a timestamp $\medianTimestamp$ independent of $\auctionid$, the content and identity remain hidden pre-release, while the adversary learns that a bid was placed for one of the $\numauctions$ (cannot distinguish which) at clock time $\medianTimestamp$.

\subsection{Phase 2: Bid-submission}\label{sec:phase2}
\paragraph{Goal.} Ensure that every \emph{timely} bid (i.e., with $\medianTimestamp\le \tend$) is considered for the auction outcome and that the \emph{authenticated winner} cannot be censored within the slot. The IL proposers verify these bids, enforce inactivity penalties. Only the winner is ultimately settled on-chain to achieve \emph{Auction Participation Efficiency (APE)}.

\paragraph{Inputs and Outputs.}
The reveal phase takes as input the user's bid transaction $\tx$ for a given auction $\auctionid$ and the timestamp certificate $(\medianTimestamp, \certificate)$ obtained during Phase~1. To make the proof for inactivity, the user inputs $\auctionid_\ell$, the last auction the user revealed for. On the IL proposer's side, the inputs are the message $(\tx,\auctionid,(\medianTimestamp,\certificate),\Pi)$ received from the user, the local record of valid bids $\setbids$, the user activity mapping $\useractivity$, and the global registries of active and slashed deposits $(\varDeposits,\variable{slashed})$.

At the end of this phase, all honest IL Proposers would have a set of valid bids $\setbids$, on top of which any auction logic can be run. The local winning bids are included in the Inclusion List for the reserved slot $S$ for the auction.

\paragraph{Overview.}
Before presenting the pseudocode in \Cref{alg:bidreveal-winner}, we provide an overview.
Each user who participated in auction $\auctionid$ submits its bid transaction $\tx$ together with its timestamp certificate $(\medianTimestamp, \certificate)$ and a proof of inactivity $\Pi$.
The proof of inactivity serves two purposes: (1) it demonstrates that, for all auctions between the user's last revealed auction $\auctionid_\ell$ and the current one, the user either revealed a valid bid or explicitly declared a no-operation ($\varNOP$); and (2) it enables IL proposers to detect users who attempt to hide activity or submit multiple bids, allowing timestampers to \emph{slash} the associated deposits.

Once received, IL proposers verify the proof, update user-activity records, and include the bid in their local set of bids \setbids. The winning bid -- according to some pre-defined auction logic -- is included in their Inclusion Lists. Only those bids, whose median timestamp $\medianTimestamp$ precedes the auction's end time $\tend$, and that are revealed to the IL Proposer before slot $S$, are considered. The inclusion of the winning bid in all honest ILs guarantees censorship resistance within the slot.

To prove that the bid presented was generated for the particular auction ID, we require the $\user$ to create another zero-knowledge proof~$\idproof$
with public elements (statement) $ \handle, \auctionid, C_\user $
and private elements (witness) $\secret_\user$,
attesting that
\begin{align*}
    \handle = \textsc{Hash}(\secret_\user \parallel &\auctionid)
    \;\; \land\;\;
     \big(C_\user = \textsc{Hash}(\secret_\user)\big).
\end{align*}
Along with this, the user must also send a history of handles it can generate for all auctions it did not participate in. This is done through a \varNOP transaction, which the IL Proposer gossips to the timestampers, so they can check for any repetition of handle. The IL Proposer then updates its local view of user's activity.

\begin{algorithm}[h!]
  \caption{Phase 2: BidReveal}
  \label{alg:bidreveal-winner}
  \footnotesize
  \begin{algorithmic}[1]
    \Require auction $\auctionid$ with end time \tend, transaction $\tx$, IL Proposers $\ILproposers$
    \Ensure IL contains $\winningbid{\auctionid} = \max(\setbids)$ for slot $S$
  \end{algorithmic}
  \algrule
  \noindent\begin{minipage}[t]{0.50\textwidth}
  \footnotesize
  \begin{algorithmic}[1]
    \Statex \textbf{Local for user $\user$:}
    \Statex \textbf{local} $\auctionid_\ell$ \rcomment{Auction last revealed for}
    \Statex \emph{At time $\clocktime \geq \tend + \gossipdelay$ run:}
    \Offline{\revealBid}{$\tx, \auctionid, (\medianTimestamp,\certificate)$}
        \State $\Pi \gets \proveInactivity(\auctionid)$
        \State $\idproof \gets \zkprove((\handle, \auctionid, C_\user), \secret_\user)$
        \For {$\ILproposer \in \ILproposers$}
        \State $(\tx, \auctionid, (\medianTimestamp,\certificate),\idproof, \Pi) \xrightarrow{\textsc{DIRECT}} \ILproposer$
        \EndFor
        \State $\auctionid_\ell \gets \auctionid$
    \EndOffline
    \vspace{4pt}
    \Offline{\proveInactivity}{$\auctionid$}
        \State $\auctionid_{max} \gets \min(\auctionid,\auctionid_\ell+3\numauctions-1)$
        \State $\Pi \gets []$
        \For {$i \in [\auctionid_\ell, \auctionid_{max}]$}
            \State $(\handle^i, \commitment^i, {\memberproof}^{,i}) \gets$$\proveEligibility(\secret_\user, \Addr_\user, i, \varNOP, \slot)$
            \State $\Pi.\mathrm{append} (i, \handle^i, \commitment^i, {\memberproof}^{,i},\varNOP )$
        \EndFor
        \State \Return $(\Pi)_{\sigma_\user}$ \rcomment{signed proof}
    \EndOffline
  \end{algorithmic}
  \end{minipage}\hfill
  \begin{minipage}[t]{0.48\textwidth}
  \footnotesize
  \begin{algorithmic}[1]
    \setcounter{ALG@line}{12}
    \Statex \textbf{Local for IL Proposer \ILproposer:}
    \Statex \textbf{local} $\setbids \gets \{\}, \useractivity \gets \{\}$ 
    \rcomment{\setbids: valid bids, \useractivity:  map from \userset to $\auctionid$s for which user has not revealed.}
    \Statex \emph{On receiving $(\tx, \auctionid, (\medianTimestamp,\certificate), \idproof, \Pi)$ from $\user$:}
    \Offline{\updateBid}{$\tx, \auctionid, (\medianTimestamp,\certificate), \idproof, \Pi$}
        \State \textbf{require} $\zkverify((\handle,\auctionid, C_\user), \idproof)$
        \State \textsc{updateHistory($\Pi$)}
        \For{$(i,\handle^i, \commitment^i, {\memberproof}^{,i},\varNOP) \in \Pi$}
         \If{$\handle^i \in \variable{slashed}$}
         \State call \textsc{slash} $(\Addr_\user, C_\user)$ on-chain, \textbf{abort}
         \EndIf
        \EndFor
        \State \textbf{require} $\Addr_\user \in \varDeposits$ \rcomment{\varDeposits at current \slot}
        \State $(\cdot,\cdot,\cdot, \tend, S) \gets$ \textsc{getAucRecord}$(\auctionid)$
        \If{$\medianTimestamp \leq \tend\;\land\;$clock time $\leq S.\text{time}$}
            \State \textbf{require} verify certificate \certificate
            \State $\setbids.\mathrm{insert}(\tx)$
        \EndIf
        \State $\useractivity[\user].\text{remove}(\auctionid)$
    \EndOffline
    
    \vspace{3pt}
    \Offline{updateHistory}{$\Addr_\user, \Pi$}
    \State Verify signature on $\Pi$
    \For{$(i,\handle^i, \commitment^i, {\memberproof}^{,i},\mathrm{NOP}) \in \Pi$}
     \State $(\handle^i, (\varNOP, \commitment^i, {\memberproof}^{,i})) \xrightarrow{\textsc{gossip}} \timestampers$
     \EndFor
     \State \textbf{wait} for $4\gossipdelay$ \rcomment{Allow \timestamper{} to slash}
     \For{$(i,\handle^i, \commitment^i, {\memberproof}^{,i},\varNOP) \in \Pi$}
         \State $\useractivity[\user]$.remove($i$)
        \EndFor
    \EndOffline
  \end{algorithmic}
  \end{minipage}
  \algrule
  \noindent\textbf{Description: }In slot $S$ of $\auctionid$, each IL Proposer runs the auction logic on \setbids and includes the winner in the local inclusion list. Users that did not reveal get $\auctionid$ added to set $\useractivity$. The proposer for a PoS blockchain only needs to include the winning bid among all local inclusion lists. This is then verified by the attesters in the PoS chain.
\end{algorithm}

\paragraph{FOCIL in our construction.}
We use a construction similar to FOCIL (\Cref{sec:focil-overview}) to obtain censorship resistance against a malicious auctioneer, with the following modifications. First, the protocol is equipped with a validity predicate, and a transaction can only be included in an IL if it satisfies the predicate. A transaction-certificate tuple $(\tx,\medianTimestamp,\certificate)$ is valid if it satisfies the following conditions:
\begin{enumerate*}
  \item All signatures in $(\certificate)$ are valid and correspond to the same commitment $(\commitment)$ and timestamp values are valid.
  \item The median timestamp \medianTimestamp is within the valid window for the current slot; specifically, the network enforces that $\medianTimestamp \leq \auctiontime$, where \auctiontime denotes the cutoff time to submit bids.
 \item The bid satisfies auction rules (e.g., valid fee, bid format).
\end{enumerate*}
If these conditions hold, an IL Proposer adds \tx to its IL, otherwise \tx is ignored. The timestamp certificate \certificate ensures that only transactions created before the cutoff \auctiontime are accepted.

Second, an IL in our construction does not contain all transactions received by its proposer but only the one -- the winning bid. This also implies that we define no limit size for an IL, as it only contains one transaction.
Finally, the rest of the parties relevant in FOCIL, i.e., validators, block builder, and attesters, are assumed to be part of underlying blockchain protocol.
We assume the following fork-choice rule: If a set of rules defined by auctions logic is not followed by the proposer on all declared inclusion lists (which are singleton for each auction ID) of setbids before, then fork choice invalidates the block (i.e., the consensus rejects the block, and the next block will be created as if previous block did not arrive).
Note that we assume that if an IL Proposer declares $\tx_1$ as its local winning bid, and another IL Proposer declares $\tx_2$, then the proposer (and attesters) must be able to identify which transaction is the correct winning bid.

\paragraph{Protocol for users.}
We now present the construction in \Cref{alg:bidreveal-winner}.
The entry point for a user is function $\revealBid()$, which sends $(\tx, \auctionid, (\medianTimestamp,\certificate), \idproof, \Pi)$ to each IL Proposer. Here, $\idproof$ is a proof that the \auctionid is the same as the one used in \handle.
$\Pi$ contains proofs for \emph{previous} auctions, which the user must submit in case it did not participate in some of them.
These proofs can be generated by function $\proveInactivity()$.

The \proveInactivity() function locally remembers the last auction $\auctionid_{\ell}$ for which the user revealed a bid. It then generates proofs for all auction numbers $i$ between $\auctionid_{\ell}$ and the current auction \auctionid (and for at most $3\numauctions - 1$ after $\auctionid_{\ell}$, which represents the maximum number of auctions the user could have received a timestamp for without being considered inactive.).
Each of those proofs is exactly the membership proof for auction $i$, as presented in \Cref{sec:timestamping}, but for a specific transaction, denoted as \varNOP, indicating that user did not participate in auction $i$.

\paragraph{Protocol for IL Proposers.}
An IL Proposer \ILproposer handles the received $(\tx, \auctionid, (\medianTimestamp,\certificate), \idproof, \Pi)$ from a user as shown in function $\updateBid()$.
The first step for \ILproposer is to forward all the proofs of inactivity $\Pi$ to the timestampers (who process them using the function $\verifyEligibility()$ of \Cref{alg:pre-phases} as described in the previous section). Then the IL Proposer waits for the timestampers to process the proofs, which may result in the user being slashed (if the timestampers detect that the user has submitted a different bid commitment for the same handle). If this is the case, the IL Proposer does not further process the received transaction.

The next step for \ILproposer is to verify the proof \idproof
and the certificate \certificate (as described previously). If successful, \ILproposer inserts the transaction \tx into its set of valid bids $\setbids$.

\paragraph{Keeping users' deposit active.}
Finally, we explain how users can make sure their deposit remains active. If a user is actively sending revealed bids to the IL Proposer, then each honest IL Proposer would have the history of handles generated by the user. This implies that the user would be active for the honest IL Proposer. A malicious IL Proposer can trigger the inactivity notice for any user; however, given the $\numauctions$ window in which the user's deposit is still active, the user can send a reveal to any honest IL Proposer that can update user's activity on-chain. The user can also prove its activity or Non Participation directly on-chain.

In case the user goes offline or sparingly participates in auctions, any IL Proposer could make the user's deposit inactive. Being inactive helps the user to not prove its entire history when it comes back online; instead it just has to prove its non-participation in $3\numauctions-1$ auctions since it went offline. This number is derived from the fact that any user can have at most $3\numauctions-1$ auctions that it has not revealed in before being added to inactives set~(\Cref{lem:max-inactivity}).

\subsection{Security Analysis}
\begin{definition}[Security assumptions]\label{def:security-assumptions}
In all the following lemmas and theorems, we assume the following:
\begin{itemize}[leftmargin=*,nosep]
  \item The network is synchronous with a delay of at most $\syncdelay$.
  \item All parties have a $\syncdelay$-synchronized clock.
  \item The gossip network has a delay of at most \gossipdelay and provides user anonymity.
  \item Static corruption with at most \ftimestampers Byzantine replicas in \timestampers, the rest being honest.
  \item Static corruption with existential honesty in \ILproposers
  \item Implicitly assume that the majority of attesters in the Proof-of-Stake chain are honest. This assumption is only required to use FOCIL as a subroutine, and not a strict requirement.
\end{itemize}
\end{definition}
Under the above assumptions, we prove the following theorem statements in \Cref{sec:security-analysis}.

\begin{restatable}[$(\delta_i,\delta_e,\tend)$-Simultaneous Release]{theorem}{censorshipresistance}
\label{thm:Simul-Release}
Assume an auction that spans \Cref{alg:pre-phases,alg:timestamp,alg:bidreveal-winner}, concludes at time $\tend$, and has some pre-defined settlement slot $S$. The auction satisfies \emph{$\delta_i$-ST Censorship Resistance} with $\delta_i = \gossipdelay+\syncdelay$ and \emph{$\delta_e$-Post-Auction Exclusion} with $\delta_e = \syncdelay$ (per \Cref{def:simulrelease}).
\end{restatable}

\vspace{-2ex}
\begin{restatable}[$\thiding-$Hiding]{theorem}{relaxedhiding}
Assume an auction that spans \Cref{alg:pre-phases,alg:timestamp,alg:bidreveal-winner}, concludes at time $\tend$. Given $\thiding = \tend+\syncdelay$ (the time at which all users can start revealing their bids), let $\textsc{AucSet}(\timestamp)$ represent the ongoing parallel auctions at clock time  $\timestamp$. Let $\userset = \{\Addr_\user:\;(\Addr_\user,\_) \in \varDeposits\}$ represent the honest users in \varDeposits in the slot right before clock time $\timestamp$. The auction satisfies the following items from the \emph{\thiding-Hiding} property (\Cref{def:hiding}):
  \begin{enumerate*}[label=(\roman*)]
      \item Indistinguishability
      \item Existential Obfuscation within $\textsc{AucSet}(\timestamp)$
      \item User Obfuscation within $\userset$.
  \end{enumerate*}
\end{restatable}

\vspace{-2ex}
\begin{restatable}[\bindingProperty]{theorem}{binding}
Assume an auction that spans \Cref{alg:pre-phases,alg:timestamp,alg:bidreveal-winner}, concludes at time $\tend$, and has some pre-defined settlement slot $S$. If $\varMinDeposit \geq (3\numauctions-1) \cdot u$, meaning that, by not revealing a commited bid, the adversary pays an amount of at least $u$ per auction, the auction satisfies the \emph{\bindingProperty} property.
\end{restatable}

\vspace{-2ex}
\begin{restatable}[Auction Participation Efficiency]{theorem}{ape} Assume an auction that spans \Cref{alg:pre-phases,alg:timestamp,alg:bidreveal-winner}, concludes at time $\tend$, and has some pre-defined settlement slot $S$. The cost to an honest active user if it does not win an auction tends to $0$.
\end{restatable}

\section{Benchmarks}

\begin{wraptable}{r}{0.5\textwidth}
\centering
\small
\setlength{\tabcolsep}{4pt}
\caption{Proof generation and verification times (in ms) for the auction and eligibility proofs.}
\label{tab:benchmarks}
\begin{tabular}{@{}l c r r@{}}
\toprule
\textbf{Proof} & \textbf{\# Bidders} & \textbf{Prove} & \textbf{Verify} \\
\midrule
\shortstack[l]{Auction $\idproof$} & --- & $12.9 \pm 1.6$ & $0.9 \pm 0.1$ \\
\midrule
\multirow{4}{*}{\shortstack[l]{Eligibility\\$\memberproof$}}
  & $2^{8}$   & $46.7 \pm 3.3$   & $0.9 \pm 0.0$ \\
  & $2^{16}$  & $72.9 \pm 3.7$   & $0.9 \pm 0.0$ \\
  & $2^{24}$  & $103.4 \pm 16.7$ & $1.0 \pm 0.1$ \\
  & $2^{32}$  & $159.4 \pm 43.2$ & $1.0 \pm 0.2$ \\
\bottomrule
\end{tabular}
\end{wraptable}

In this section, we evaluate the computational cost of the two zero-knowledge
proofs that constitute the core of our protocol: the eligibility proof
$\memberproof$ from Phase~1 (\Cref{sec:timestamping}) and the auction proof
$\idproof$ from Phase~2 (\Cref{sec:phase2}).
We implement both proofs using the Groth16~\cite{DBLP:conf/eurocrypt/Groth16} proof system over the BN254 curve,
as implemented in Rust in the \texttt{arkworks} library,\footnote{\url{https://arkworks.rs}}
and employ the Poseidon hash function~\cite{DBLP:conf/uss/0001KR0S21} for all hash computations.
For the eligibility proof $\memberproof$ we vary the depth of the Merkle tree to support different total numbers of bidders, ranging from $2^{8}$ to $2^{32}$.
For each configuration we report the mean and standard deviation of the proof generation and verification times over 10~iterations with independently sampled random inputs.
The benchmarks were run on a Macbook with the Apple M3 CPU and 24 GB of RAM. The repository for the proofs is \href{https://github.com/OrestisAlpos/proofs-censorship-resistant-auctions}{\underline{github.com/OrestisAlpos/proofs-censorship-resistant-auctions}}.

The auction proof $\idproof$ is lightweight, requiring approx. only $13$\,ms to generate
and $0.9$\,ms to verify.
Even in the most demanding configuration, with a Merkle tree supporting over
four billion bidders, the eligibility proof $\memberproof$ can be generated in
under $160$\,ms and verified in $1$\,ms.

\section{Discussions}

\paragraph{Pure hiding vs.\ cryptographic hiding.}
A distinctive feature of our hiding construction is that it is \emph{purely information-theoretic}: bid values are hidden behind a hash commitment and never encrypted under any shared or threshold key.
This contrasts with threshold-encryption-based approaches (e.g., encrypted mempools~\cite{shutter24,choudhuri2025practical,boneh2025batch}), where hiding relies on a cryptographic threshold: if the decryption committee colludes or is compromised, every submitted bid is \emph{decrypted}, revealing its full content to the adversary.
In our protocol, the timestamping committee is trusted only for \emph{liveness}: if a majority of timestampers collude, the worst outcome is that a user's bid does not receive a valid timestamp and thus cannot participate.
The bid value itself is never at risk of exposure, since the commitment $\commitment = \textsc{Hash}(\tx \parallel C_\user)$ is information-theoretically hiding (given the randomness in $C_\user$) and is never opened to the timestampers.
In other words, a corrupted committee degrades to a \emph{denial-of-service} rather than a \emph{privacy breach} -- a strictly preferable failure mode for high-value auctions.
This is an added advantage over works like~\cite{sealedbid25} which rely on HECC in order to ensure hiding.

\paragraph{Limitations of existing anonymous channels.} We stress that the anonymous broadcast channel is an idealization; no deployed system simultaneously achieves strong sender anonymity, low latency, and a deterministic delivery bound.
Tor-based overlays~\cite{tor} offer reasonable anonymity but add multi-second latency and provide no worst-case delivery guarantee.
Mix networks~\cite{chaum1981,nym} offer stronger anonymity but incur higher latency and are vulnerable to intersection attacks under low traffic.
Protocols with provable anonymity such as Dissent~\cite{dissent} or Riposte~\cite{riposte} achieve strong guarantees but only for small groups and incur latency measured in seconds to minutes.
The choice of instantiation introduces a tradeoff between hiding strength and the parameter~$\gossipdelay$: a lower-latency overlay reduces $\gossipdelay$ but weakens the \emph{User Obfuscation} and \emph{Existential Obfuscation} guarantees (while \emph{Indistinguishability} remains unaffected, since it relies on the commitment scheme and ZK proofs).
For Ethereum-based auctions with 12\,s slots, a gossip delay of $\gossipdelay \in [1\text{\,s},\,4\text{\,s}]$ (achievable with Tor-like overlays or lightweight mix networks) leaves ample time within a slot for both phases.
We view the design of anonymous gossip layers with formal latency bounds as complementary future work.

\paragraph{Strength of obfuscation.}
When $|\textsc{AucSet}| = 1$ or only one user is registered for the auction, existential obfuscation is trivial -- the adversary knows which auction the bid belongs to and/or who the user is.
This is inherent: when only one auction is active, submitting a bid commitment reveals participation regardless of the protocol.
Even then, our definition still guarantees \emph{Indistinguishability} (bid value hidden) and \emph{User Obfuscation} (bidder identity hidden among registered users), which together are strictly stronger than prior commit-and-reveal designs.
In practice, high-value blockchain auctions tend to overlap temporally, making $|\textsc{AucSet}| \gg 1$ the common case.

\section{Related Work}

\begin{table}[t!]
\centering
\caption{Comparison of sealed-bid auction protocols against the four properties. \ding{51}: achieved, \ding{55}: not achieved, $\sim$: partially achieved or with caveats.}
\label{tab:related-work}
\footnotesize
\begin{tabular}{l c c c c}
\toprule
\textbf{Protocol} & \textbf{Hiding} & \textbf{Simul.\ Release} & \textbf{No Free Withdraw} & \textbf{APE} \\
\midrule
Commit-and-reveal~\cite{ens} & $\sim$ & \ding{55} & \ding{55} & \ding{55} \\
FAST~\cite{david2022fast} & $\sim$ & \ding{55} & \ding{51} & \ding{55} \\
Riggs~\cite{Riggs} & $\sim$ & \ding{55} & \ding{51} & \ding{55} \\
Cryptobazaar~\cite{novakovic2024cryptobazaar} & $\sim$ & \ding{55} & \ding{51} & \ding{55} \\
ZeroAuction~\cite{zeroauction} & $\sim$ & \ding{55} & \ding{51}& \ding{55} \\
MCP~\cite{sealedbid25} & $\sim$\textsuperscript{a} & \ding{55}\textsuperscript{b} & \ding{51} & \ding{55} \\
\textbf{This work} & \ding{51}\textsuperscript{c} & \ding{51} & \ding{51} & \ding{51} \\
\bottomrule
\multicolumn{5}{l}{\textsuperscript{a}\footnotesize Value indistinguishability via HECC, but no existential obfuscation or user obfuscation.}\\
\multicolumn{5}{l}{\textsuperscript{b}\footnotesize Single-slot only; multi-slot auctions are vulnerable to cross-slot censorship.}\\
\multicolumn{5}{l}{\textsuperscript{c}\footnotesize Existential obfuscation within $\textsc{AucSet}$; see \Cref{rem:ideal-hiding} for the ideal strengthening.}
\end{tabular}
\end{table}

Auctions are a foundational economic tool, and blockchains provide a novel but adversarial environment for their implementation. We survey prior designs relevant to sealed bids, off-chain execution, and protocol-level resource allocation.

\begin{description}[leftmargin=0pt, labelindent=0pt]
\item\textbf{Open On-Chain Auctions.} Open auctions, particularly English and Dutch formats, are widely deployed on-chain due to their simplicity. English auctions~\cite{milgrom1982theory} suffer from frontrunning and sniping because bids are publicly visible, prompting mitigations such as randomized or extended closing times, e.g., Polkadot's candle auction~\cite{polkadotcandle}. Dutch auctions~\cite{vickrey1961counterspeculation} reduce interaction and settle faster; MakerDAO adopted them for liquidations~\cite{makerdao}. Nevertheless, public bid visibility fundamentally limits ST censorship resistance, even with ordering or timing mitigations.

\item\textbf{Commit-and-Reveal Sealed-Bid Auctions.} Sealed-bid auctions are harder to implement on-chain due to transparency. The standard commit-and-reveal approach hides bids during commitment and reveals them later, enabling formats such as Vickrey auctions~\cite{vickrey1961counterspeculation}. Early deployments, including ENS auctions~\cite{ens}, demonstrated feasibility but also exposed drawbacks such as latency, inclusion fees for losing bidders, and reveal-phase manipulation, leading to subsequent redesigns~\cite{ens1,Braghin2018,KST+20,LZL+21,CLXW22}.
Commit-and-reveal introduces multiple rounds, operational overhead, and incentives to delay or withhold reveals. Even with commitments, known deadlines induce bid clustering and grant proposers temporary inclusion power, enabling selective censorship or timing-based bias. Our work avoids these issues by eliminating on-chain reveals and preventing proposers from influencing bid inclusion.

\item\textbf{Cryptographic and TEE-Based Designs.} Several proposals enhance sealed-bid auctions using cryptography or trusted execution. Encrypted mempools hide transaction contents until inclusion~\cite{fernando2025trx,choudhuri2025practical,boneh2025batch,choudhuri2024mempool,agarwal2025efficiently,bormet2025beat}. FAST~\cite{david2022fast} implements sealed-bid auctions using confidential transactions and penalties, assuming underlying privacy support.
Other systems rely on semi-trusted auctioneers with verifiable computation~\cite{Abulkasim2021,constantinides2021block,desai2019hybrid,galal2018verifiable,sharma2021anonymous}, trusted hardware enclaves~\cite{desai2021secauctee}, or multi-party computation~\cite{Blass2018}. Riggs~\cite{Riggs} employs time-lock puzzles to ensure eventual revelation. Cryptobazaar~\cite{novakovic2024cryptobazaar} and ZeroAuction~\cite{zeroauction} reduce on-chain leakage via encrypted bids and delayed decryption, but still incur ST censorship risk and inclusion costs for losing bids.
These designs improve privacy but typically rely on additional trust assumptions or coordination mechanisms. Our approach avoids trusted parties and records only the winning bid on-chain while mitigating ST censorship.

\item\textbf{Off-Chain Auctions.} Many systems move auctions off-chain to reduce costs. NFT marketplaces such as OpenSea rely on off-chain signed bids~\cite{opensea}, while CoW Protocol uses an off-chain orderbook with competing solvers~\cite{cow}. At the infrastructure level, priority gas auctions~\cite{daian2020flashboys} evolved into private off-chain auctions via Flashbots~\cite{flashbots}, and later into Proposer-Builder Separation (PBS) implemented through MEV-Boost~\cite{mevboost}, effectively treating blockspace as a recurring first-price auction. While off-chain auctions scale efficiently, they depend on intermediaries that can censor or delay bids. Our design removes this reliance while retaining auction efficiency.

\item\textbf{Protocol-Level Censorship Resistance.} Protocol mechanisms such as FOCIL~\cite{soispoke2024focil,eip-7805} introduce committee-based inclusion lists to raise censorship costs but rely on existential honesty and remain vulnerable to block-filling. AUCIL~\cite{wadhwa2025aucil} strengthens these guarantees in a rational model using auctions over inclusion lists, but still grants proposers a latency advantage. Multiple Concurrent Proposers (MCP)~\cite{sealedbid25} provide stronger ST censorship resistance via architectural changes, at the cost of higher complexity and deployability constraints. Our work occupies a middle ground, offering auction-level ST censorship resistance without modifying the underlying consensus protocol.

\item\textbf{Rational Auction Protocols.} A related line of work studies sealed-bid auctions under rational adversaries. Ganesh et al.~\cite{ganesh2022secure,ganesh2024secure} design cryptographic protocols for first-price and Vickrey auctions that ensure incentive compatibility via punishments and commitments. These works focus on strategic correctness among rational bidders, but do not address blockchain-specific challenges such as ST censorship or timing leakage. In contrast, our construction targets permissionless deployment and directly mitigates inclusion-time censorship and bid-existence leakage using timestamp certificates and inclusion lists, while minimizing losing-bidder fees through our APE objective.
\end{description}

\section{Conclusion and Future Work}\label{sec:discussion}

Our protocol fits into a broader effort to make sealed-bid auctions practical on blockchains.
Conventional sealed-bid designs do not address ST censorship: because the block proposer has monopoly control for the duration of its slot, a malicious proposer can suppress timely honest bids or exploit its latency advantage to react strategically to information that other bidders were not able to observe.
More recent approaches remove the proposer monopoly by employing multiple concurrent proposers, but doing so demands substantial modifications to consensus, introduces considerable protocol complexity, and imposes significant on-chain load.

Our work takes a different direction.
We describe a lightweight, consensus-adjacent mechanism that achieves the strengthened properties
required for high-value, time-sensitive auctions. In particular, our protocol provides (i) $(\gossipdelay+\syncdelay,\,\syncdelay,\,\tend)$-simultaneous release, (ii) $\thiding$-hiding, (iii) no free bid withdrawal, and (iv) auction participation efficiency. Importantly, because most of the mechanism operates off-chain, it places minimal demand on blockspace, though this necessitates a more involved deposit mechanism and eligibility proofs.

\paragraph{Future Work.}
First, a key direction for future work is designing incentive mechanisms for timestampers. Our protocol currently relies on an honest majority among timestampers to produce correct and timely certificates; replacing this assumption with explicit economic incentives would provide stronger robustness and make the timestamping layer secure even under strategic behavior. A related question exists for IL proposers, though it is less pressing in our setting -- since we require only existential honesty -- and prior work has already made progress on incentivizing IL proposers~\cite{wadhwa2025aucil}.

Second, our protocol does not yet include a Sybil-resistance mechanism for auctioneers to prevent malicious parties from spamming many low-value auctions. This could be addressed through a reservation fee mechanism for slot reservations, which would also allow auctioneers to explicitly reserve capacity only for the slots in which their auctions run. More generally, fee or deposit-based eligibility rules, or the simpler option of treating auctioneers as a permissioned set, remain important directions for future work.

Finally, a full economic analysis remains open. This includes determining appropriate penalties for non-reveals, incentives for block proposers to include auction-settlement transactions, and -- if timestampers are made rational rather than honest -- the corresponding incentive structure for correct and timely timestamping.

\bibliography{refs}

\appendix

\section{Security Analysis}\label{sec:security-analysis}

\subsection{Intermediate results}
We first prove some intermediate results, which we then use to prove the main theorems for our construction.

\begin{lemma}[Median robustness of timestamp certificate]\label{lem:median-robustness}
Let $\timestampers$ be a committee of timestampers with at most $\ftimestampers$ corrupted, where $|\timestampers| = 2 \ftimestampers + 1$. Let $\mathcal{T}=\{\timestamp_i : i\in \timestampers\}$ be the set of timestamps that the user collects on the commitment $\commitment$,
where \emph{at most} $\ftimestampers$ timestamps of $\mathcal{T}$ may be adversarially chosen.
Given $\gossipdelay$ as the maximum delay of the gossip network and $\syncdelay$-synchronized clocks, and assuming that the user treats any missing stamps at timeout as $\infty$.
Let $\medianTimestamp$ be the median of $\mathcal{T}$, and $\mathcal{T}_{\mathsf{hon}}$ represents the subset received from honest parties.
Then:
\begin{enumerate}[leftmargin=*]
  \item \textbf{Honest span.}
  \[
    \min (\mathcal{T}_{\mathsf{hon}}) \le \medianTimestamp \le \max (\mathcal{T}_{\mathsf{hon}}),
  \]
  In particular, the adversary cannot move the median outside the honest span.
  \item \textbf{User drift bound.} If $t_\user$ represents the time at which user sends its transaction
  \[
    \max (\mathcal{T}_{\mathsf{hon}}) - t_\user \le \gossipdelay+\syncdelay,
  \]
  \[
   \min (\mathcal{T}_{\mathsf{hon}})-t_\user \ge -\syncdelay,
  \]
  \item \textbf{Drift bound.}
  \[
    \max (\mathcal{T}_{\mathsf{hon}}) - \min (\mathcal{T}_{\mathsf{hon}}) \le 2\syncdelay+\gossipdelay,
  \]
  hence $\medianTimestamp$ deviates from any honest stamp by at most 2$\syncdelay+\gossipdelay$.

\end{enumerate}
\end{lemma}

    \begin{proof}[Proof. Honest span]
        We prove this statement by Pigeon Hole Principle. There exist $\ftimestampers$ values before and after the median timestamp. Since we have $\ftimestampers+1$ honest timestamps, either the median timestamp is honest, or at least one of the timestamps before or after the median is filled with an honest timestamp. Thus $\min(\mathcal{T}_{\mathsf{hon}}) \leq \medianTimestamp$ and $\max(\mathcal{T}_{\mathsf{hon}}) \geq \medianTimestamp$.
    \end{proof}
    \begin{proof}[User drift bound]
        If the user sends the transaction at their local clocktime $\clocktime_\user$, then from the perspective of an honest timestamper ($\timestamper{i}$), the sending time of the transaction (accounting for the clock drift) must be \[\clocktime_\user-\syncdelay \leq \clocktime_\user^{\timestamper{i}} \leq \clocktime_\user+\syncdelay.\]
        The time it receives the transaction and assigns it a timestamp is thus \[\clocktime_\user-\syncdelay \leq \timestamp_{i} = (\clocktime_\user^{\timestamper{i}} +\gossipdelay) \leq \clocktime_\user+\syncdelay+\gossipdelay.\] Thus, for all honest $\timestamper{i}$, the assigned timestamp is given by the previous equation.
        In other words,
        \begin{equation}
            \label{eq:userdriftupperbound}
            \max(\mathcal{T}_{\mathsf{hon}})\leq \clocktime_\user+\syncdelay+\gossipdelay
        \end{equation}
        and
        \begin{equation}
        \label{eq:userdriftlowerbound}
        \min(\mathcal{T}_{\mathsf{hon}}) \geq \clocktime_\user-\syncdelay
        \end{equation}
        which proves the statement.
    \end{proof}
    \begin{proof}[Drift bound]
        Subtracting \Cref{eq:userdriftupperbound,eq:userdriftlowerbound}, we get
        \begin{equation*}
        \label{eq:driftbound}
        \max(\mathcal{T}_{\mathsf{hon}})-\min(\mathcal{T}_{\mathsf{hon}}) \leq 2\syncdelay+\gossipdelay
        \end{equation*}
    \end{proof}

\begin{lemma}[Completeness of membership proofs]\label{lem:phaseB-completeness}
  An honest user can always create a statement $(\rootslot, \handle, \commitment)$ and verifying membership proof \memberproof for any auction $\auctionid$.
\end{lemma}
\begin{proof}
  An honest user \user registers by choosing $\secret_\user$ and computing $C_\user$,
  which then gets added in the \varDeposits variable of \depositcontract (see $\registerDeposit()$ in \Cref{alg:pre-phases}).
  The deposit $(\Addr_\user, C_\user)$ can only be removed from \varDeposits through function $\makeOffline()$, $\textsc{slash}()$ or $\textsc{withdraw}()$ (see \Cref{alg:pre-phases}).

  For an honest \user, no honest IL Proposer will call $\makeOffline(\Addr_\user)$,  because \user reveals its bid commitments or participates in all auctions using a \varNOP bid. A malicious IL Proposer \emph{can} call $\makeOffline(\Addr_\user)$ against honest \user, however, this gives the honest user (or any other honest IL Proposer) a $\numauctions$ window in which it can use $\updateActivity()$ to avoid becoming inactive.

  Moreover, $\textsc{slash}()$ or $\textsc{withdraw}()$ can only be called if either the user uses a $\handle$ twice or withdraws its own funds, which the honest does not do.

  Hence, deposits $(\Addr_\user, C_\user)$ of honest \user will remain in the \varDeposits. This means that \user can create the \varPath variable from \rootslot (the root of the merkle tree that contains all deposits in \varDeposits) to $(\Addr_\user, C_\user)$ for some slot $\slot$ (see $\proveEligibility()$ in \Cref{alg:timestamp}).
  Moreover, \user can create the handle $\handle = \textsc{Hash}(\secret_\user \parallel \auctionid)$ for any auction $\auctionid$ and commitment $\commitment = \textsc{Hash}(\tx \parallel C_\user)$ for its transaction \tx.
  Hence, \user can compute a valid witness $(\secret_\user,\Addr_\user,\varPath, \tx)$ for the statement $(\rootslot, \handle, \commitment)$ and, from the completeness property of the zero-knowledge proof system, a verifying membership proof \memberproof for any auction $\auctionid$.
\end{proof}

\begin{lemma}[Max Inactivity]
    \label{lem:max-inactivity}
    A user can get a timestamp certificate for at most $n_c = 3\numauctions-1$ auctions without revealing the corresponding bids.
\end{lemma}
\begin{proof}
    Let a user's last revealed auction id be $\auctionid_\ell$. If the user does not participate for $\numauctions$ auctions, at least honest IL Proposer will call $\textsc{makeOffline}()$. This starts a $\numauctions$-auction countdown, for which the user remains active (and thus can get timestamp certificates). When the last of these auctions is about to end, there can be at most $\numauctions-1$ more auctions started and still be running. A valid certificate can be generated for all of these $3\numauctions-1$ auctions; however, as soon as the last auction ends, the user's deposit becomes inactive and no other $\handle$ can be generated with a valid $\memberproof$.
\end{proof}

\begin{lemma}[Uniqueness of handle per deposit]\label{lem:phaseB-uniqueness}
  For each deposit (as indicated by $C_\user$) and auction $\auctionid$, the user can create at most one handle $\handle$, such that $\memberproof$ is a verifying membership proof for the statement $(\rootslot, \handle, \commitment)$.
\end{lemma}
\begin{proof}
  Assuming that $\textsc{Hash}()$ is a collision-resistant hash function, a user \user cannot come up with $\secret'_\user \neq \secret_\user$, such that $\textsc{Hash}(\secret'_\user) = C'_\user$ and $C'_\user = C_\user$.
  Hence, \user would have to use $\secret'$ to create a valid membership proof \memberproof for a statement $(\rootslot, \handle', \commitment')$, where $\handle' = \textsc{Hash}(\secret'_\user \parallel \auctionid)$ and $\commitment' = \textsc{Hash}(\tx' \parallel C'_\user)$, for an entry $(\Addr_\user, C'_\user)$ that does not exist in the merkle tree \rootslot. This is impossible, except with negligible probability, because it would contradict the collision resistance of $\textsc{Hash}()$, used to implement the Merkle tree.
\end{proof}

\subsection{Proofs for the Entire Construction}

\censorshipresistance*

\begin{proof}
(Censorship Resistance) For every honest user $\user \in \userset$, let $\tx$ be a bid transaction created by $\user$ for auction $\auctionid$ and sent to timestampers before clock time $\tend - \delta_i$ (i.e., $\clocktime_\user \leq \tend-(\gossipdelay+\syncdelay)$). From \Cref{lem:phaseB-completeness}, the honest user can always create such a bid.

By \Cref{lem:median-robustness}, \emph{User Drift Bound}, we know that
\[
\max(\mathcal{T}_{\mathsf{hon}})\leq \clocktime_\user+\syncdelay+\gossipdelay.
\]
Also from \Cref{lem:median-robustness}, \emph{Honest Span}, we get that
\[
\medianTimestamp \leq \max(\mathcal{T}_{\mathsf{hon}}) \leq \clocktime_\user+\syncdelay+\gossipdelay \leq \tend,
\]
i.e., the median timestamp obtained has value at most $\tend$. Thus, if such a bid transaction with timestamp is sent to \ILproposers, it is still considered valid.
Then, $\tx$ is included in the input bid set of auction $\auctionid$ used by all honest IL Proposers when computing the auction outcome.

(Post Auction Exclusion) From \cref{lem:median-robustness}, if a transaction is generated by an adversary at time $t_u$, then the minimum timestamp it can receive is $\timestamp\geq\min(\mathcal{T}_{\mathsf{hon}})\geq t_u-\syncdelay$

If the bid is valid, i.e., $\timestamp \leq \tend$ then $t_u \leq \tend+\syncdelay$. Thus, no adversarial bids sent after $\syncdelay$ of auction end $\tend$ is valid, satisfying post-auction exclusion.
\end{proof}

\relaxedhiding*
\begin{proof} [Proof. Indistinguishability]
    Let us consider, towards a contradiction, that an auction protocol violates indistinguishability. We will show that using this auction's adversary as a subroutine, we can design an adversary that has an advantage in the decision game for the pre-image of a hash function or learning witness of a zero knowledge proof.

    Concretely, $\adversary$ outputs two transactions $(\tx_0,\tx_1)$ and is given a
    challenge commitment $\commitment_b = \Hash(\tx_b \parallel C_\user)$
    together with a ZK membership proof created as
    \[
      {\memberproof}_{,b} \gets \zkprove((\rootslot,\handle,\commitment_b),(\secret_\user,\Addr_\user,\varPath,\tx_b)),
    \]
    where $b\gets \{0,1\}$ is drawn uniformly at random.
    Before time $\thiding$, all information available to~$\adversary$ consists of:
    \begin{itemize}[leftmargin=*]
      \item the handle $\handle=\Hash(\secret_\user\parallel\auctionid)$,
      \item the commitment $\commitment_b$,
      \item the proof ${\memberproof}_{,b}$, and
      \item the timestamp certificate $(\medianTimestamp,\certificate)$,
    \end{itemize}
    none of which reveals $\tx_b$ directly.
    i) Since $\handle$ is independent of $\tx_b$, there is no advantage that the adversary gains from this information. By hash function's pre-image resistance, $\secret_\user$ and $\auctionid$ remain hidden with no advantage more than random guessing.
    ii) the timestamp certificate $(\medianTimestamp,\certificate)$ contains only timestamps and signatures. Again no information about $\tx_b$ is available here.

    By assumption, the adversary outputs a guess~$b'$ such that
    \[
      \left|\Pr[b'=b] - \tfrac{1}{2}\right| \geq \epsilon(\lambda)
    \]
    for some non-neglibile function $\epsilon$.

    Now, for the commitment $\commitment_b = \Hash(\tx_b \parallel C_\user)$, consider an adversary $\adversary_1(\adversary)$, such that,
    when $\adversary$ outputs two messages $(\tx_0,\tx_1)$, the adversary $\adversary_1$ outputs $(m_0 = (\tx_0\parallel C_\user), m_1 =(\tx_1 \parallel C_\user))$. The challenger to the pre-image resistance function returns $\Hash(m_b)$ and this value is passed to $\adversary$. When $\adversary$ outputs $b$, $\adversary_1$ outputs the same $b$.
    Since $\adversary_1$ must not have any advantage (by definition of $\Hash()$ being pre-image resistant), $\adversary$ can also not obtain any advantage through the commitment $\commitment_b$.

    Lastly, for the proof ${\memberproof}_{,b}$, consider an adversary $\adversary_2(\adversary)$. When $\adversary$ outputs two messages $(\tx_0,\tx_1)$, adversary $\adversary_2$
    chooses a random $b \gets \{0,1\}$ and constructs the corresponding
    commitment
    \[
        \commitment_b = \Hash(\tx_b \parallel C_\user).
    \]
    It then submits the ZK statement $(\rootslot,\handle,\commitment_b)$ to
    its ZK challenger, receiving either a real or simulated proof
    $\pi^*$ depending on the hidden bit $\beta$.
    $\adversary_2$ now continues the protocol with
    ${\memberproof}_{,b} = \pi^*$ and sends to $\adversary$ the tuple
    \[ (\handle,\commitment_b,{\memberproof}_{,b},(\medianTimestamp,\certificate)).
    \]
    Adversary $\adversary$ outputs $b'$. If $b = b'$, $\adversary_2$ outputs $\beta = $ simulated. Otherwise, output $\beta = $ real.
    However, since $\adversary_2$ must have no advantage in the zero-knowledge game, $\adversary$ cannot have an advantage in the indistinguishability game.
\end{proof}
\begin{proof} [Existential Obfuscation within $\textsc{AucSet}(\medianTimestamp)$]
   Let us consider, towards a contradiction, that this advantage is not neglible.
     Before time $\thiding$, all information available to~$\adversary$ consists of:
    \begin{itemize}[leftmargin=*]
      \item the handle $\handle^b=\Hash(\secret_\user\parallel\auctionid_b)$,
      \item the commitment $\commitment = \Hash(\tx \parallel C_\user)$,
      \item the proof ${\memberproof}_{,b}$, and
      \item the timestamp certificate $(\medianTimestamp,\certificate)$,
    \end{itemize}
    The proof follows the same arguments as \emph{indistinguishability}. Only information  $\handle^b=\Hash(\secret_\user\parallel\auctionid_b)$ and the proof ${\memberproof}_{,b}$ are dependent on the auction ID. The timestamping-certificate time $\medianTimestamp$ only reveals that $\auctionid_b \in \textsc{AucSet}(\medianTimestamp)$ and thus the adversary $\adversary$ can be used to build adversaries that have an advantage in learning a pre-image of a hash function or learning witness of a zero knowledge proof.
\end{proof}
\begin{proof} [User obfuscation within $\userset$]
    Before time $\thiding$, all information available to~$\adversary$ consists of:
    \begin{itemize}[leftmargin=*]
      \item the handle $\handle^b=\Hash(\secret_b\parallel\auctionid)$,
      \item the commitment $\commitment_b = \Hash(\tx \parallel C_b)$,
      \item the proof ${\memberproof}_{,b}$, and
      \item the timestamp certificate $(\medianTimestamp,\certificate)$,
    \end{itemize}
    Here, the timestamp certificate $(\medianTimestamp,\certificate)$ reveals only that whoever submitted the transaction is in $\userset$. The other three elements are all similar to previous proofs: if advantage is obtained from either the handle $\handle^b$ or $\commitment_b$, then we can make an adversary with an advantage in pre-image of a hash function, and if the advantage is obtained from proof ${\memberproof}_{,b}$, then we can make an adversary with an advantage in learning witness of a zero knowledge proof.
\end{proof}

\ape*
\begin{proof} [Informal Proof]
    If the user is active and honest, then after depositing, the user never needs to make an on-chain call. This can be done by revealing its bids for auctions it participates in, and by submitting a proof of non-participation immediately if it does not. Even if a malicious IL Proposer calls $\makeOffline()$, an honest IL proposer with its history will always call $\updateActivity()$ before the user is ever made inactive.
    Thus, the only time the user pays a fee is when it registers its deposit, and thus the amortized cost over many auctions that the user participates but does not win in, will tend to $0$.
\end{proof}

\binding*
\begin{proof}
  
  Consider that a bidder decides not to reveal a bid $\tx$ after receiving a timestamp certificate for some auction id $\auctionid$. The handle for this auction is unique, as proved in~\Cref{lem:phaseB-uniqueness}, and was used to get the timestamp for $\tx$. Thus, creating a proof of non-participation for \handle would double-use $\handle$, causing $\textsc{slash}()$ to be called on-chain. Thus, this would imply that the user can never generate a proof of non-participation for such an $\auctionid$. Thus, it can never withdraw its deposit from the contract loosing its deposit.

  However, the bidder can do so for at most $n_c = 3\numauctions-1$ auctions (from~\Cref{lem:max-inactivity}) and thus the amortized cost for not revealing any valid winning bid is at least $\tfrac{\varMinDeposit}{n_c}$, which is at least $u$, per the theorem statement.
\end{proof}

\end{document}